\documentclass[11pt]{amsart}
\usepackage{amssymb}
\usepackage{amsthm}
\usepackage{amsmath}
\usepackage[dvips]{epsfig}
\usepackage{graphicx}
\usepackage{color}
\usepackage{url}
\usepackage{epstopdf}
\usepackage{caption}
\usepackage{subfig}
\usepackage{tikz}
\usetikzlibrary{decorations.pathmorphing}

\newtheorem{theorem}{\bf Theorem}[section]

\newtheorem{corollary}{\bf Corollary}[section]

\newtheorem{definition}{\bf Definition}[section]
\newtheorem{lemma}{\bf Lemma}[section]

\newcommand{\RE}{{\rm Re}}
\newcommand{\diff}[2]{\frac{d #1}{d #2}}

\newcommand{\IM}{{\rm Im}}
\newcommand{\ARG}{{\rm arg}}

\newcommand{\eps}{\varepsilon}

\makeatletter\@addtoreset {equation}{section}\makeatother

\begin{document}
\title[Lattice solitons in DNLS equation with saturation]{\bf Standing lattice solitons in the discrete NLS equation with saturation}

\author{G.L. Alfimov}
\address[G.L. Alfimov]{National Research University of Electronic Technology MIET, Zelenograd, Moscow 124498, Russia}
\address[G.L. Alfimov]{Institute of Mathematics, Ufa Research Center, Chernyshevskii~str. 112, Ufa 450008, Russia}

\author{A.S. Korobeinikov}
\address[A.S. Korobeinikov]{National Research University of Electronic Technology MIET, Zelenograd, Moscow 124498, Russia}

\author{C.J.  Lustri}
\address[C.J.  Lustri]{Department of Mathematics, Macquarie University, Sydney, NSW, Australia}

\author{D.E. Pelinovsky}
\address[D.E. Pelinovsky]{Department of Mathematics, McMaster University, Hamilton, Ontario, L8S 4K1, Canada}
\address[D.E. Pelinovsky]{Department of Applied Mathematics, Nizhny Novgorod State Technical University, 603950, Russia}



\maketitle

\begin{abstract}
We consider standing lattice solitons for discrete nonlinear Schr\"{o}dinger equation with saturation (NLSS), where
so-called {\it transparent points} were recently discovered. These transparent points are the values of
the governing parameter (e.g., the lattice spacing) for which the  Peierls--Nabarro barrier vanishes.
In order to explain the existence of transparent points, we study a solitary wave solution in
the continuous NLSS and analyse the singularities of its analytic continuation in the complex plane.
The existence of a quadruplet of logarithmic singularities nearest to the real axis is proven
and applied to two settings: (i) the fourth-order differential equation arising as the next-order
continuum approximation of the discrete NLSS and (ii) the advance-delay version of the discrete NLSS.

In the context of (i), the fourth-order differential equation generally does not have solitary wave solutions 
due to small oscillatory tails. Nevertheless, we show that solitary waves solutions exist for specific values
of governing parameter that form an infinite sequence. We present an asymptotic formula for the distance
between two subsequent elements of the sequence in terms of the small parameter of lattice spacing.
To derive this formula, we used two different analytical techniques: the semi-classical limit of
oscillatory integrals and the beyond-all-order asymptotic expansions. Both produced the same result
that is in excellent agreement with our numerical data.

In the context of (ii), we also derive an asymptotic formula for values of lattice spacing
for which approximate standing lattice solitons can be constructed. The asymptotic formula
is in excellent agreement with the numerical approximations of transparent points. However, we show that
the asymptotic formulas for the cases (i) and (ii) are essentially different and that
the transparent points do not generally imply existence of continuous standing lattice solitons
in the advance-delay version of the discrete NLSS.
\end{abstract}

\newpage

\section{Introduction}
\label{intro}

Lattice differential equations in the form of the discrete nonlinear Schr\"{o}dinger (NLS) equations
are commonly met in applications since they express the leading-order balance between
the nonlinear and periodic properties of many physical systems \cite{pelin-book}. Lattice solitons
represent elementary excitations in nonlinear lattices which appear naturally in many physical experiments \cite{panos-book}.

Since continuous translational invariance is broken in the lattice differential equations,
travelling waves do not usually propagate steadily. Instead they slow down and stop near a particular lattice site.
The related Peierls--Nabarro (PN) energy barrier is the energy difference between two pinned lattice solitons,
one of which is symmetric about a lattice site and the other one is symmetric about the midpoint between two nearest
lattices sites. The two families of standing lattice solitons can be pinned to any lattice site thanks to the discrete translational invariance of the lattice differential equations.

The cubic discrete NLS equation has the two pinned standing lattice solitons \cite{QinXiao}
and exhibit no other single-humped solutions at least for sufficiently small values of lattice spacing
\cite{PelRothos}. In the past few years, there have been many attempts
to construct generalizations of the cubic discrete NLS equation, which have continuous families
of standing lattice solitons \cite{dmitriev1,dmitriev2,pel1} (see also \cite{pel2} for travelling
lattice solitons in the same models). Such continuous families are parameterized by
the spatial translation parameter which provides a continuous deformation between the two pinned lattice solitons.
The PN energy barrier is identically zero for the continuous families of standing lattice solitons.
The main problem of the discrete NLS models exhibiting continuous families of standing lattice solitons
is that these models do not typically arise in physical applications.

One possible generalization of the cubic NLS equation arising in many optical applications
is the NLS equation with saturation (NLSS) \cite{GH1991}. With a suitable normalization,
the discrete NLSS is written as the following lattice differential equation for the sequence of complex amplitudes
$\{ \psi_n(t) \}_{n \in \mathbb{Z}} \in \mathbb{C}^{\mathbb{Z}}$ evolving in time $t \in \mathbb{R}$:
\begin{eqnarray}
i \frac{d \psi_n}{dt} + \frac{1}{h^2} (\psi_{n+1} - 2 \psi_n + \psi_{n-1}) + \psi_n -\frac{\theta \psi_n}{1+|\psi_n|^2} = 0,
\quad n \in \mathbb{Z},
\label{dNLS}
\end{eqnarray}
where $h$ is the lattice spacing parameter and $\theta$ is the saturation parameter.
When the saturable nonlinearity is expanded in power series and the quintic and higher-order powers are truncated,
one can obtain the cubic discrete NLS equation for the amplitude $\phi_n(t) = \psi_n(t) e^{i (\theta - 1) t}$:
\begin{eqnarray}
i \frac{d \phi_n}{dt} + \frac{1}{h^2} (\phi_{n+1} - 2 \phi_n + \phi_{n-1}) + \theta |\phi_n|^2 \phi_n = 0, \quad
\quad n \in \mathbb{Z},
\label{dNLS-cubic}
\end{eqnarray}
which is focusing if $\theta > 0$.

Numerical studies of the discrete NLSS showed existence of
standing lattice solitons with zero PN energy barrier \cite{Maulskov,Melvin1}
as well as existence of travelling lattice solitons  \cite{Melvin1,Melvin2,barashenkov2007}.
It was observed in \cite{Melvin1,Melvin2} that standing lattice solitons with zero PN energy barrier
exist for a set of points with respect to a governing parameter (called {\em transparent points}),
whereas the travelling lattice solitons exist on a set of bifurcation curves in the velocity-frequency
parameter plane. It was conjectured in \cite{Melvin1} that
the sequence of such transparent points or bifurcation curves is unbounded, although
the numerical results only captured the first few transparent points or bifurcation curves.
More recent numerical studies \cite{Syafwan} showed stability of travelling lattice solitons
in the discrete NLSS.

The purpose of this work is to explain the phenomenon of a countable sequence of
transparent points for standing lattice solitons in the discrete NLSS. Standing
lattice solitons satisfy the following second-order difference equation:
\begin{gather}\label{DiscreteU}
\frac1{h^2} (u_{n+1}-2u_n+u_{n-1}) + u_n-\frac{\theta u_n}{1+u_n^2}=0, \quad n \in \mathbb{Z}.
\end{gather}
Two particular solutions to the difference equation (\ref{DiscreteU}) are generally known \cite{QinXiao}:
on-site soliton $\{u_n^{os}\}_{n \in \mathbb{Z}}$ and inter-site soliton $\{u_n^{is}\}_{n \in \mathbb{Z}}$,
according to the following symmetry conditions:
\begin{equation}
\label{on-site-inter-site}
u_{-n}^{os}=u_{n}^{os}, \quad u_{-n}^{is}=u_{n-1}^{is}, \quad n \in \mathbb{Z}.
\end{equation}
Both lattice solitons decay to zero as $|n| \to \infty$ and the transparent point is the value
of $h$ (for fixed $\theta$) for which the PN energy barrier vanishes \cite{Maulskov,Melvin1}\footnote{It was
shown in \cite{Melvin1} that the energy of the lattice soliton must be modified
by the mass term in order to get correct conclusions on the PN energy barrier compared
to the earlier work \cite{Maulskov}.}. It was shown in \cite{Melvin1} that zeros of
the PN energy barrier occur roughly at the values of $h$ for which linearization
of the difference equation (\ref{DiscreteU}) at the on-site and inter-site solitons
(\ref{on-site-inter-site}) admits zero eigenvalue. Interchange between stability of the on-site
and inter-site solitons in the time-evolution problem (\ref{dNLS}) occur at these values of $h$,
although the two sets are not necessary the same. Thanks to these observations,
we adopt the following definition of the transparent points in the discrete NLSS.

\begin{definition}
\label{def-1}
We say that $h=h_0^{os}$ (or $h=h_0^{is}$) is the transparent point of the difference equation (\ref{DiscreteU})
at the on-site (inter-site) soliton satisfying (\ref{on-site-inter-site}) if
the Jacobian operator at the corresponding soliton admits a zero eigenvalue.
\end{definition}

Continuous generalization of the difference equation (\ref{DiscreteU}) is the following advance-delay equation:
\begin{eqnarray}
\frac{1}{h^2} \left[ u(x+h) - 2 u(x) + u(x-h) \right] + u(x) - \frac{\theta u(x)}{1+u(x)^2}=0, \quad x \in \mathbb{R}.
\label{advance-delay-intro}
\end{eqnarray}
On-site and inter-site discrete solitons satisfying (\ref{on-site-inter-site}) do not generally
correspond to continuous solutions to the advance-delay equation (\ref{advance-delay-intro}). Indeed,
the difference equation (\ref{DiscreteU}) is formulated as a two-dimensional discrete map with the saddle zero equilibrium;
stable and unstable manifolds of this equilibrium intersect generally at a {\it discrete} set.
Therefore, unless the two manifolds coincide like in the two-dimensional discrete maps considered in \cite{HPS,P11},
no continuous standing lattice solitons exist in the advance-delay equation (\ref{advance-delay-intro})  at a transparent point $h$
of Definition \ref{def-1}.

The only exception from the general observation above is the point $h = h_1 := \sqrt{2}$,
for which an exact solution $u \in C(\mathbb{R})$ exists because the advance-delay equation (\ref{advance-delay-intro})
with $h = \sqrt{2}$ corresponds to the integrable Ablowitz--Ladik lattice with a large class of exact solutions \cite{Khare2005}.
This particular value of $h$ was found in \cite{Melvin1,Melvin2} to be the first one
in the sequence $\{ h_m \}_{m \in \mathbb{N}}$ of the numerically detected transparent points of Definition \ref{def-1}.
The prediction of the transparent points was confirmed by direct numerical simulation of the lattice solitons
in the discrete NLS equation (\ref{dNLS}) that exhibited radiationless propagation \cite{Melvin1,Melvin2}.

In order to explain the numerical results in \cite{Melvin1,Melvin2}, we analyze
the following second-order differential equation:
\begin{eqnarray}
\frac{d^2 u}{dx^2} + u -\frac{\theta u}{1+u^2}=0,
\label{e=0-intro}
\end{eqnarray}
which is the formal limit of the advance-delay equation (\ref{advance-delay-intro}) as $h \to 0$.
A solitary wave solution decaying to zero at infinity
exists for every $\theta > 1$.  We extend the solution analytically off the real line and
prove that the nearest singularities in the analytic continuation of solutions are located
symmetrically as a quadruplet in the complex plane. The following theorem represents
the main result of this analysis.

\begin{theorem}
\label{theorem-main}
For every $\theta > 1$, there exists a unique positive and decaying solution
$U \in C^{\infty}(\mathbb{R})$ to the second-order equation (\ref{e=0-intro})
which is continued analytically off the real line until the nearest singularities
at $\pm \alpha \pm i \beta$, where $\alpha, \beta > 0$. For every $z \in \mathbb{C}$
close to $z_0 = -\alpha + i \beta$ with $\arg(z_0-z) \in \left(-\frac{\pi}{2},\frac{3\pi}{2}\right)$, the solution $U$
satisfies
\begin{equation}
\label{singular-behavior-intro}
U(z) = i + \sqrt{\theta} (z - z_0) \sqrt{\log(z_0-z)} \left[ 1 +
\mathcal{O}\left(\frac{\log|\log|z-z_0||}{\log|z-z_0|}\right) \right] \quad \mbox{\rm as} \quad z \to z_0,
\end{equation}
whereas the behavior of $U$ at other singularity points is obtained from the symmetry conditions
\begin{eqnarray}
U(\bar z)=\overline{U(z)},\quad U(-z)=U(z), \quad z \in \mathbb{C}.
\label{Symm-intro}
\end{eqnarray}
\end{theorem}

Theorem \ref{theorem-main} is applied to the study of solitary wave solutions
in the following fourth-order differential equation
\begin{eqnarray}
\varepsilon^2 \frac{d^4 u}{d x^4} + \frac{d^2 u}{dx^2} + u -\frac{\theta u}{1+u^2}=0,
\label{ode-intro}
\end{eqnarray}
where $\varepsilon$ is a small parameter. The fourth-order equation (\ref{ode-intro})
arises from the advance-delay equation (\ref{advance-delay-intro}) in the next order
to the second-order equation (\ref{e=0-intro}) thanks to
the formal power expansion:
\begin{equation}
\label{expansion-difference}
u \in C^{\infty}(\mathbb{R}) : \quad
\frac{1}{h^2} \left[ u(x+h) - 2 u(x) + u(x-h) \right] = \frac{d^2 u}{dx^2} + \frac{h^2}{12} \frac{d^4 u}{d x^4} + \mathcal{O}(h^4),
\end{equation}
with the correspondence $\varepsilon := h/(2 \sqrt{3})$.
A solitary wave solution decaying to zero at infinity does not typically
exist in the fourth-order equation (\ref{ode-intro}) because of exponentially small oscillatory tails at infinity
\cite{GrimshawJoshi,PRG}. Exponential asymptotic expansions (also known as {\em beyond-all-order asymptotics})
were developed to analyze these exponentially small oscillatory tails both
for the differential equations \cite{Tovbis1,Tovbis2}, advance-delay equations of the Henon type \cite{Tovbis},
and the differential advance-delay equations \cite{Iooss,CMP2009,barashenkov2007}.

In the method of beyond-all-order asymptotics, the existence of solitary waves decaying to zero at infinity
can be justified by computations of a scalar function called {\em the Stokes constant}.
In many cases, the Stokes constant is either nonzero \cite{Tovbis1,Tovbis}
or vanishes on a finite set of isolated points of the one-parameter line \cite{CMP2009}.
This situation occurs typically in the case when the analytic continuation of the solitary wave solution
has a pair of symmetric singularities nearest to the real line. It was realized some time ago \cite{Gel1,Gel2} that
if the analytic continuation of the solitary wave solution has a quadruplet of symmetric singularities nearest
to the real line, the oscillations on the solution's tail may be suppressed at a countable set of isolated points
on the one-parameter line.

This phenomenon was recently studied in the context of the lattice differential equations.
By analyzing oscillatory integrals in the semi-classical limit, the very similar explanation for
the onset of a countable sequence of travelling lattice solitons was proposed and illustrated for a number
of physically relevant examples including the Klein--Gordon lattice with
the cubic--quintic nonlinearity \cite{prl-alfimov}.
By analyzing the beyond-all-order asymptotics, travelling lattice solitons in
the diatomic Fermi--Pasta--Ulam lattice were explained similarly in \cite{Lustri}.
These travelling lattice solitons arises as a result of co-dimension one bifurcations
among more general travelling solutions with exponentially small oscillatory tails \cite{Hoffman,Lustri,Vainstein}.

In the present work, we demonstrate analytically and numerically
that the symmetric location of the branch point singularities
in the solitary wave solution to the second-order equation (\ref{e=0-intro}) explains
the onset of a countable sequence of co-dimension one bifurcations
for the solitary waves of the fourth-order equation (\ref{ode-intro}) with
$\varepsilon$ near $\{ \varepsilon_m \}_{m \in \mathbb{N}}$. The sequence
for the lattice spacings $\{ h_m \}_{m \in \mathbb{N}}$ with $h_m = 2 \sqrt{3} \varepsilon_m$
accumulates to zero as $m \to \infty$ according to the asymptotic representation:
\begin{equation}
\label{spacing-odeintro}
h_m \sim \frac{4 \sqrt{3} \alpha}{\pi (2m-1)}, \quad m \in \mathbb{N},
\end{equation}
where $\alpha > 0$ is a numerical parameter in Theorem \ref{theorem-main}.

Compared to the previous works in \cite{prl-alfimov,Lustri}, the technical challenge
of our work is caused by the fact that the solitary wave solution to the second-order equation (\ref{e=0-intro})
is not available in the closed analytical form. Another challenge is that the asymptotic behavior
involves the logarithmic singularity.
We show that both analytical techniques developed independently in \cite{prl-alfimov,Lustri}
lead to the same predictions for the fourth-order equation (\ref{ode-intro}).

One can anticipate a similar sequence $\{ h_m \}_{m \in \mathbb{N}}$
to arise in the advance-delay equation (\ref{advance-delay-intro}), for which the
fourth-order equation (\ref{ode-intro}) is the first-order approximation.
Indeed, we show existence of a countable sequence $\{ h_m \}_{m \in \mathbb{N}}$,
for which the first Stokes constant vanishes in the advance-delay equation (\ref{advance-delay-intro}).
However, there are two important differences between predictions for the advance-delay equation (\ref{advance-delay-intro})
and the fourth-order equation (\ref{ode-intro}). First,
the sequence $\{ h_m \}_{m \in \mathbb{N}}$ accumulates to zero as
$m \to \infty$ according to a different asymptotic representation:
\begin{equation}
\label{eps-advance-delay-intro}
h_m \sim \frac{4 \alpha}{(2m-1)}, \quad m \in \mathbb{N},
\end{equation}
where $\alpha > 0$ is the same as in (\ref{spacing-odeintro}).
The reason for the discrepancy is a different dispersion relation between the advance-delay equation
(\ref{advance-delay-intro}) and the fourth-order equation (\ref{ode-intro}).

Second and mostly important, no existence of continuous solution $u \in C(\mathbb{R})$
to the advance-delay equation (\ref{advance-delay-intro}) with $h$
near $\{ h_m \}_{m \in \mathbb{N}}$ can be demonstrated because there are infinitely many
resonant roots of the dispersion relation in (\ref{advance-delay-intro})
compared to only one root in (\ref{ode-intro}). This corresponds to the necessity
of checking infinitely many Stokes constants for computations of standing
lattice solitons. The result (\ref{eps-advance-delay-intro}) is deduced
from the first Stokes constant, whereas all others are expected to be
nonzero near $\{ h_m \}_{m \in \mathbb{N}}$ with the exception of $h_1 = \sqrt{2}$,
for which the exact solution exists and ensures that all Stokes constants vanish simultaneously.
Therefore, our results for the advance-delay equation (\ref{advance-delay-intro})
only allow us to predict {\em an approximate standing lattice soliton
with a single hump at the center and the smallest oscillatory tails in the far-field}.
Such approximate standing lattice solitons arise roughly at same values of $h$
corresponding to the transparent points in Definition \ref{def-1}.

The countable sequence of transparent points is related to the phenomenon of snaking
of standing lattice solitons discussed for the cubic--quintic discrete NLS equation
in \cite{Chong1,Chong2} and for the Allen--Cahn lattice in \cite{Dawes}. Indeed, the snaking
is induced by the existence of two countable sequences of instability bifurcations for on-site and
inter-site lattice solitons, which are typically located at different points in the governing
parameter (see Figure 3 in \cite{Chong2}). At each instability bifurcation, two branches of
either on-site or inter-site lattice solitons merge in a fold bifurcation, where they exchange their stabilities.
In addition, asymmetric lattice solitons bifurcate from the same fold points and connect branches of
the on-site and inter-site solitons. If each branch of asymmetric lattice solitons existed
at the {\em same} point of the instability bifurcation for the limiting on-site and inter-site
solitons, this would suggest the existence of continuous solutions to the advance-delay
equation at this point. However, the asymmetric lattice solitons are typically
connected to the on-site and inter-site solitons at {\em different} points. As a result, the transparent points
do not guarantee bifurcations of continuous solutions in the advance-delay equation.

The paper is organized as follows. Section 2 is devoted to analysis of singularities
in the second-order equation (\ref{e=0-intro}) and gives the proof of Theorem \ref{theorem-main}.
Validity of the asymptotic formula (\ref{spacing-odeintro}) for the fourth-order equation (\ref{ode-intro})
is shown in Section 3 analytically and numerically. Section 4 reports analogous
results for validity of the asymptotic formula (\ref{eps-advance-delay-intro})
for the advance-delay equation (\ref{advance-delay-intro}).
Section 5 concludes the paper with a summary.

\section{Solitary wave solution to the second-order equation}
\label{sec-2}

Here we study the second-order differential equation:
\begin{eqnarray}
\frac{d^2 u}{dx^2} + u -\frac{\theta u}{1+u^2}=0
\label{e=0}
\end{eqnarray}
where $\theta$ is the model parameter. Solutions to the second-order equation (\ref{e=0})
can be obtained from the first-order invariant
\begin{equation}
\label{first-order}
E = \left( \frac{du}{dx} \right)^2 + u^2 - \theta \log(1 + u^2),
\end{equation}
where the value of $E$ is a constant in $x$. The implicit formula for a solution to
the initial-value problem
\begin{eqnarray*}
u(x_0)=u_0,\quad \frac{du}{dx}(x_0)=\sqrt{E+\theta\log(1+u_0^2)-u_0^2}
\end{eqnarray*}
is given by
\begin{eqnarray}
x-x_0=\int_{u_0}^{u}\frac{d\xi}{\sqrt{E+\theta \log(1+\xi^2)-\xi^2}}.
\label{RealSol}
\end{eqnarray}
Solitary wave solutions satisfy the decay conditions
$u(x)\to 0$ as $x \to \pm \infty$
and correspond to the level $E=0$.
The exponential decaying solutions
exist in (\ref{e=0}) if $\theta > 1$.

Zeros of the denominator in (\ref{RealSol}) with $E = 0$ determine the turning points for
the second-order equation (\ref{e=0}).
In particular, the real root of transcendental equation
\begin{eqnarray}
\theta \log(1+u^2)-u^2 = 0
\label{Trans_eq}
\end{eqnarray}
corresponds to the maximum value of the solitary wave. Complex roots are
important for analytic continuation of the solitary wave solutions into the complex plane. Hence, we define the function
\begin{equation}
\label{function-f}
f(u) := \theta \log(1+u^2)-u^2.
\end{equation}
and  analyze its real and complex zeros in Section \ref{Zeros}.
Solitary wave solutions and their analytic continuations in the complex plane
are studied in Section \ref{Soliton_e=0}. Asymptotic properties
of the analytic continuation and the proof of Theorem \ref{theorem-main}
are given in Section \ref{AsPropU}.

\subsection{Zeros of the function $f(u)$}\label{Zeros}

Let us start with the following lemma:
\begin{lemma}
\label{lemma-0}
For every $\theta > 1$ there exists only one positive root of
the nonlinear equation (\ref{Trans_eq}) denoted by $u^*$, and this root is simple.
\end{lemma}

\begin{proof}
Consider roots of $f(u) : \mathbb{R}^+ \mapsto \mathbb{R}$. We have $f(0) = f'(0) = 0$, $f'(\sqrt{\theta-1}) = 0$,
$$
\left\{ \begin{array}{l} f'(u) > 0, \quad u \in (0,\sqrt{\theta-1}), \\
f'(u) < 0, \quad u \in (\sqrt{\theta-1},\infty), \end{array} \right.
$$
and $\lim_{u \to +\infty} f(u) = -\infty$. This implies that there is
exactly one root of $f(u) : \mathbb{R}^+ \mapsto \mathbb{R}$ denoted by $u^*$.
It is straightforward to check that $u^* \neq \sqrt{\theta - 1}$ for any $\theta > 1$
so that the root $u^*$ is simple.
\end{proof}

\begin{corollary}
\label{cor-0}
For every $\theta > 1$, the nonlinear equation (\ref{Trans_eq}) has only three real roots given by
two simple roots $\pm u^*$ and the double root at $0$.
\end{corollary}

\begin{proof}
It is obvious that $0$ is also a root of $f(u) : \mathbb{R} \mapsto \mathbb{R}$, and it is a double root.
By the symmetry $f(-u) = f(u)$, there exists also a simple root at $-u^*$.
\end{proof}

Consider now the function $f(u)$ for $u$ complex. The function $f(u)$ has two branch points at $u=\pm i$.
We restrict the consideration by one sheet of the Riemann surface of $f(u)$. Specifically, we consider $f(u)$
on the set $P$ defined as the entire complex plane with two horizontal cuts $\mathcal{Q}_1= (-\infty+ i, i]$,
$\mathcal{Q}_2= (-\infty- i, -i]$. The function $f(u)$ is holomorphic in $P$ and has in $P$
at least three zeros $u=u^*$, $u=-u^*$ and $u=0$. The following theorem states that $f(u)$ has no other zeros in $P$.

\begin{theorem}
\label{theorem-00}
For every $\theta > 1$, the nonlinear equation (\ref{Trans_eq}) has only three roots in $P$ given by
two simple roots $\pm u^*$ and the double root at $0$.
\end{theorem}

In order to prove Theorem \ref{theorem-00} we need the following technical lemma,
the proof of which is a straightforward exercise.

\begin{lemma}
\label{lemma-00}
Assume that $\rho>0$ is small enough and $\theta>1$. Let $\gamma_u$ be the path in the complex plane
shown in Fig.~\ref{Cut_map}(a).
Assume that $f(u)$ is represented on ${\bf u}^-$  by  the main branch of the logarithm,
$\log u=\log|u|+i~\ARG~ u$ for $\ARG~u \in(-\pi,\pi)$.
Then the function $w=f(u)$ maps $\gamma_u$ into $f(\gamma_u)=f({\bf u}^+)\cup f(C^\rho_i) \cup f({\bf u}^-)$
shown in Fig.~\ref{Cut_map}(b), where
\begin{itemize}
\item[(i)] $f({\bf u}^-)$ is a $U$-shape curve given in the parametric form by
\begin{eqnarray}
f({\bf u}^-) = \theta \log(1 + (i-t)^2) - (i-t)^2, \quad t \geq \rho.
\label{W_x_y}
\end{eqnarray}
$f({\bf u}^-)$ intersects the real axis once, at some point of {\em positive} semi-axis;

\item[(ii)] $f({\bf u}^+)$ is a copy of $f({\bf u}^-)$ shifted by $2\pi\theta i$.
$f({\bf u}^+)$ does not cross the real axis;

\item[(iii)] $f(C^\rho_i)$ is a path that connects the endpoint of $f({\bf u}^-)$ corresponding to $t=\rho$ with the corresponding endpoint of $f({\bf u}^+)$.  $f(C^\rho_i)$ crosses the real axis once, at some point of negative semi-axis.
\end{itemize}
\end{lemma}

\begin{figure}[h]
{\centerline{\includegraphics [scale=0.6]{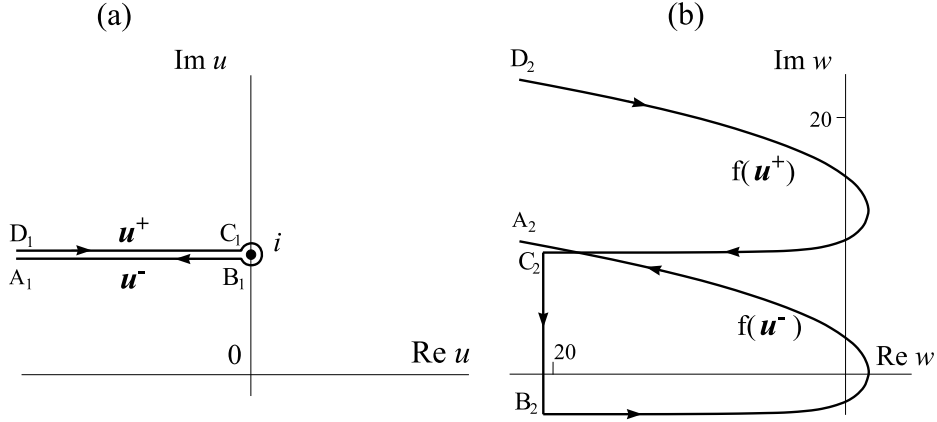}}}
\caption{The path $\gamma_u$ (a) and its image $f(\gamma_u)$ (b). The plot for (b) was computed numerically
for $f(u)$ with $\theta=5$. The points $A_2,\ldots,D_2$ are images of $A_1,\ldots,D_1$.}
\label{Cut_map}
\end{figure}

\noindent {\em Proof of Theorem \ref{theorem-00}.} Consider the contour $\Gamma$ shown in Fig.~\ref{Arg_principle}.
We assume that $R$ is large enough and $\rho>0$ is arbitrarily small. The argument principle states that
the number of zeros of $f(u)$ (taking into account their multiplicity) within $\Gamma$  is equal to
the number of turns around the origin that makes $f(u)$ when $u$ goes around $\Gamma$.
Due to symmetry $u\to \bar{u}$ of the contour $\Gamma$ and since $f(-R)=f(R)$ the numbers of
turns of $f(u)$ are equal for the two parts of $\Gamma$ situated in the upper and lower half-planes.

Consider the part of $\Gamma$ in upper half-plane between the points $u=R$ and $u=-R$.
Along this part of $\Gamma$, $f(u)$ makes one complete turn clockwise when passing along the big semi-circle $|u|=R$
and, due to Lemma \ref{lemma-00}, one more complete turn when getting round the cut $\mathcal{Q}_1$.
Therefore the total number of turns of $f(u)$ for $\Gamma$ is equal to 4. However $f(u)$ has already
three zeros within $\Gamma$: the simple zeros $u=u^*$, $u=-u^*$ and the double zero $u=0$. Therefore
$f(u)$ has no other zeros in $\Gamma$. Since $R$ is arbitrarily large and $\rho$ is arbitrarily small,
we arrive at the desired result. $\Box$

\begin{figure}[h]
{\centerline{\includegraphics [scale=0.6]{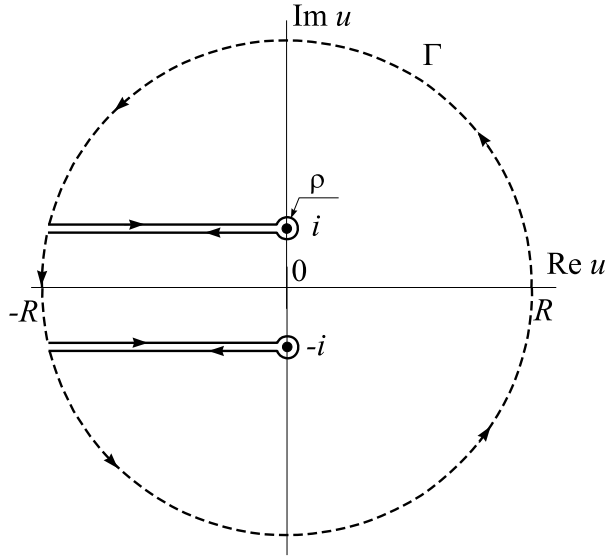}}}
\caption{The contour $\Gamma$ for the proof of Theorem \ref{theorem-00}.} \label{Arg_principle}
\end{figure}

\subsection{Analytical continuation of the solitary wave solution}
\label{Soliton_e=0}

Simple analysis of the phase plane $(u,u')$ for the second-order equation (\ref{e=0}) yields the following.
For $\theta>1$, $(0,0)$ is a saddle point on the phase plane $(u,u')$ with eigenvalues $\lambda=\pm\sqrt{\theta-1}$.
The solitary wave solution corresponds to the homoclinic loop of this equilibrium. Due to Lemma \ref{lemma-0},
there exists unique (up to the involution $u\to-u$) symmetric homoclinic loop of $(0,0)$. Therefore there exists
unique (up to transformation $u\to -u$) even solution $U(x)$ such that $U(x) > 0$ for every $x\in\mathbb{R}$
and $U(x) \to 0$ as $|x| \to \infty$. The function $U(x)$ can be written in an implicit form as follows:
\begin{eqnarray}
\int_{u^*}^U \frac{du}{\sqrt{\theta\log(1+u^2)-u^2}} = -|x|, \quad x \in \mathbb{R},
\label{exact}
\end{eqnarray}
where $u^*$ is the unique positive root in Lemma \ref{lemma-0}. The solution $U(x)$ decays to zero exponentially fast as $|x| \to \infty$,
\begin{eqnarray}
U(x)\sim C e^{-\sqrt{\theta-1} |x|},\quad |x|\to \infty\label{ExpAsym_U}
\end{eqnarray}
where $C$ is a constant that depends on $\theta$ only.
Two profiles of the solution $U(x)$  are presented in Fig.~\ref{U_profile} for $\theta = 2$ and $\theta = 5$.

\begin{figure}[h]
{\centerline{\includegraphics [scale=0.6]{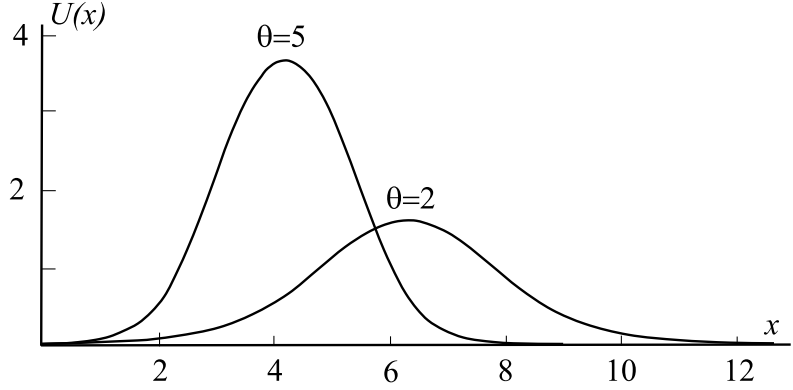}}} \caption{Plots of the function $U(x)$ versus real $x$
for $\theta=2$ and $\theta=5$} \label{U_profile}
\end{figure}


Denote the analytic continuation of $U(x)$ into the complex plane $z = x + i y$ by $U(z)$.
Since $U$ is real and even on the real axis, then $U$ in the complex plane satisfies the conditions
\begin{eqnarray}
U(\bar z)=\overline{U(z)},\quad U(-z)=U(z).
\label{Symm}
\end{eqnarray}
Implicit formula for $U(z)$ is obtained from formula (\ref{exact}) as follows:
\begin{eqnarray}
z(U) = \int_{\gamma} \frac{du}{\sqrt{\theta\log(1+u^2)-u^2}},
\label{exact-complex}
\end{eqnarray}
where $\gamma$ is a path that connects the points $u=u^* \in \mathbb{R}$ and $u=U \in \mathbb{C}$ in $P$ that
does not cross the branch cuts $\mathcal{Q}_1$ and $\mathcal{Q}_2$.
We choose in (\ref{exact-complex}) the branch for the square root such that $\sqrt{r e^{i\phi}}=\sqrt{r}e^{i\phi/2}$
for $\phi\in(-\pi,\pi)$. The integrand has a pole at $u=0$ and square root branching points at $u=\pm u^*$.
We introduce one more cut along the real axis, $\mathcal{Q}_3=(-\infty, u^*]$, and define
the set $Q$ on Fig.~\ref{Tilde_Q}(a).

\begin{lemma}\label{lemma-symm}
Let $z(U)$ be defined for $U \in Q$ by formula (\ref{exact-complex}). Then
\begin{eqnarray}
z(\bar U)=-\overline{z(U)}\label{z-symmetry}
\end{eqnarray}
\end{lemma}

\begin{proof}
Consider the points $U\in Q$ and $\bar U\in Q$. Link $U$ and $u^*$ with
some path $\gamma$ in $Q$ and consider $z(U)$ defined by (\ref{exact-complex})
with this $\gamma$. Link $\bar U$ and $u^*$ with the path $\bar \gamma$ that is symmetric to $\gamma$
with respect to the real axis and consider $z(\bar U)$ defined by (\ref{exact-complex}) with this $\bar \gamma$.
In small vicinity of $u^*$ the path $\gamma$ has the parametrization $u=u^*+re^{i\phi(r)}+\mathcal{O}(r^2)$
and the path $\bar \gamma$ has the parametrization $u=u^*+re^{-i\phi(r)}+\mathcal{O}(r^2)$. Then in this vicinity of $u^*$
\begin{eqnarray*}
\sqrt{f(u)}&=&\sqrt{f'(u^*)re^{i\phi(r)}+\mathcal{O}(r^2)}\quad \mbox{at}\quad \gamma\\
\sqrt{f(u)}&=&\sqrt{f'(u^*)re^{-i\phi(r)}+\mathcal{O}(r^2)}\quad \mbox{at}\quad \bar\gamma
\end{eqnarray*}
Note that $f'(u^*)<0$ for all $\theta>1$. This implies that the signs of $\sqrt{f(u)}$
are opposite on the pathes $\gamma$ and $\bar \gamma$ (otherwise, the function $\sqrt{f(u)}$ defined
in vicinity of $u^*$ in $Q$ has a discontinuity on the real axis). This proves the symmetry formula (\ref{z-symmetry}).
\end{proof}

\begin{figure}[h]
{\centerline{\includegraphics [scale=0.7]{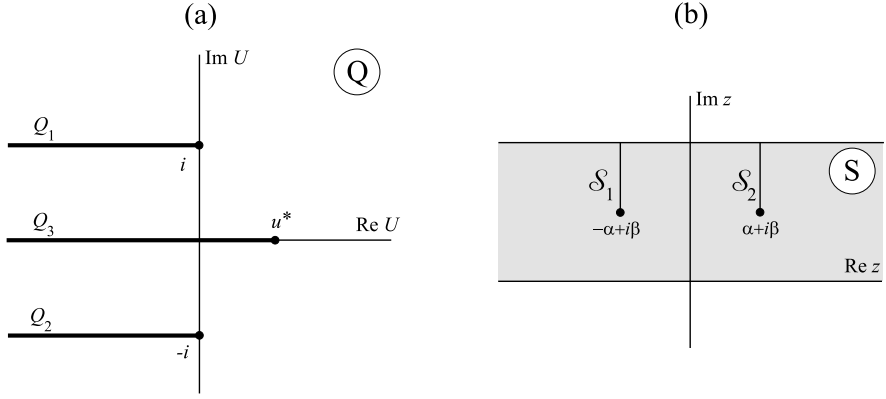}}}
\caption{The set $Q$ (a) and the strip $S$ (b).}
\label{Tilde_Q}
\end{figure}

Now we are in position to prove the following result.

\begin{theorem}\label{theorem-001}
The function $z(U)$ given by (\ref{exact-complex}) defines a conformal mapping of $Q$ such that:
\begin{itemize}
\item[(a)] For $r>0$
\begin{eqnarray}
z(r)=\pm\frac{\log r}{\sqrt{\theta-1}}+\mathcal{O}(1),\quad r\to 0, \label{Small_r01}
\end{eqnarray}
where ``$+$'' and ``$-$'' correspond to upper and lower edge of $\mathcal{Q}_3$ respectively, and
\begin{eqnarray}
z\left(re^{i\phi}\right)-z(r)=\frac{i|\phi|}{\sqrt{\theta-1}}+\mathcal{O}(r),
\quad \phi\in(-\pi,\pi)\label{Small_r02}
\end{eqnarray}
\item[(b)]  If $U\in(0,u^*)$ for the upper and lower edges of $\mathcal{Q}_3$, then
\begin{equation}
z(-U)=z(U)+\frac{\pi i}{\sqrt{\theta-1}};
\label{Symmetry_U}
\end{equation}
\item[(c)] The points $U=\pm i$ map into the points $z=\mp\alpha+i\beta$
where $\alpha,\beta > 0$ are given by
\begin{equation}
\label{alpha-beta}
\alpha = J_1 + J_2, \quad \beta = \frac{\pi}{2\sqrt{\theta-1}},
\end{equation}
with
\small
\begin{eqnarray*}
J_1 & := & \int_{1}^{u^*}\frac{du}{\sqrt{\theta \log(1+u^2)-u^2}}, \\
J_2 & := & -\int_0^1 \frac{\theta(\log(1+u^2) + \log(1-u^2))}
{\sqrt{\theta \log(1+u^2)-u^2}\sqrt{-\theta \log(1-u^2)-u^2}(\sqrt{\theta \log(1+u^2)-u^2} + \sqrt{-\theta \log(1-u^2)-u^2})}.
\end{eqnarray*}
\normalsize
\item[(d)] the image of $Q$ shown on Fig.\ref{Tilde_Q}(b) includes the set $S$ that consists of the strip $\{0\leq\IM~z\leq 2\beta\}$
with two vertical cuts $\mathcal{S}_1=[-\alpha+i\beta,-\alpha+2i\beta]$ and $\mathcal{S}_2=[\alpha+i\beta,\alpha+2i\beta]$.
\end{itemize}
\end{theorem}

\begin{proof}
By Theorem \ref{theorem-00} the denominator of the integrand in (\ref{exact-complex}) has
no zeros in the interior of $Q$. Therefore the integrand is holomorphic in the interior of $Q$
and the result of integration in (\ref{exact-complex}) does not depend on $\gamma$. Hence,
the function $z(U)$ in (\ref{exact-complex}) is also holomorphic in $Q$.\\

\begin{figure}[h]
\centerline{\includegraphics [scale=0.75]{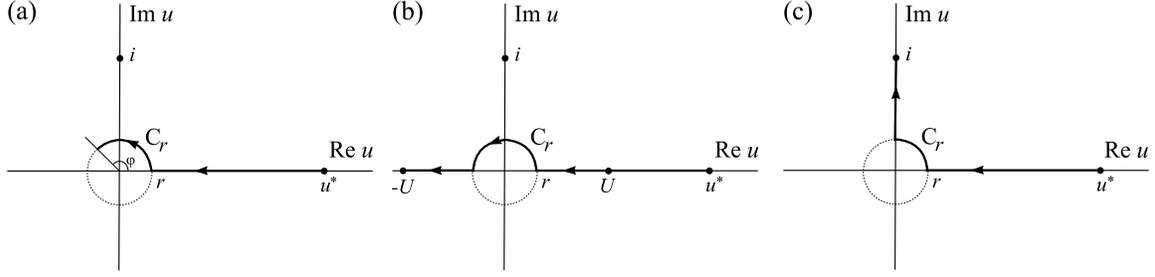}}
\caption{Pathes of integration in the proof of Theorem \ref{theorem-001} used
for formula (\ref{Small_r02}) (a), formula (\ref{Symmetry_U}) (b), and formula (\ref{alpha-beta}) (c).} \label{Path}
\end{figure}

\noindent {\em Proof of (a).} Formula (\ref{Small_r01}) follows immediately from formula (\ref{ExpAsym_U}).
In order to prove formula (\ref{Small_r02}) assume that $\phi\in(0,\pi)$ and
consider the path shown in Fig.~\ref{Path}(a). Let $u=re^{i\phi}$ be the parametrization on $C_r$.
We have
\begin{eqnarray*}
z\left(re^{i\phi}\right)=z(r)+\int_0^\phi\frac{ire^{i\phi}}{\sqrt{(\theta-1)r^2e^{2i\phi}-\theta r^4e^{4i\phi}/2+\mathcal{O}(r^4)}}=
z(r)+\frac{i\phi}{\sqrt{\theta-1}}+\mathcal{O}(r)
\end{eqnarray*}
that yields (\ref{Small_r02}). The same formula arises for $\phi\in(-\pi,0)$ if the symmetry (\ref{z-symmetry}) is used. $\Box$ \\

\noindent {\em Proof of (b).}  Note that for every $\theta > 1$ and $|\IM~u|<1$ the following representation holds,
\begin{equation}
\label{representation-D}
\sqrt{\theta\log(1+u^2)-u^2} = u \Phi_{\theta}(u),
\end{equation}
where $\Phi_{\theta}$ is a holomorphic function of $u$. In the circle $|u|<1$ the function $\Phi_{\theta}(u)$ can be represented
by the following Taylor series:
\begin{equation*}
\Phi_{\theta}(u)=\sqrt{\theta-1}\left[ 1 + \frac{\theta}{\theta-1} \sum_{k=1}^{\infty} \frac{(-1)^k u^{2k}}{k+1} \right]^{1/2}.
\end{equation*}
Therefore, $\Phi_{\theta}(u)$ is an even function and $\Phi_{\theta}(0) =\sqrt{\theta - 1} > 0$.
Consider the path in Fig.~\ref{Path}(b) that connects the points $u=-U$ and $u=u^*$,
passes along the upper edge of the cut $\mathcal{Q}_3$, and includes $C_r$,
the arc of the circle of radius $r > 0$ situated in the upper half of the complex plane.
We obtain
\begin{eqnarray*}
& \phantom{t} & \int_{u^*}^{-U}\frac{du}{\sqrt{\theta\log(1+u^2)-u^2}} - \int_{u^*}^{U}\frac{du}{\sqrt{\theta\log(1+u^2)-u^2}} \\
& = & \int_{U}^{-U}\frac{du}{\sqrt{\theta\log(1+u^2)-u^2}}=
\left( \int_{U}^{r} + \int_{C_r} + \int_{-r}^{-U} \right) \frac{du}{\sqrt{\theta\log(1+u^2)-u^2}}.
\end{eqnarray*}
By the representation (\ref{representation-D}), the third integral is equivalent to
$$
\int_{-r}^{-U} \frac{du}{\sqrt{\theta\log(1+u^2)-u^2}} =
\int_{r}^{U} \frac{du}{\sqrt{\theta\log(1+u^2)-u^2}} = - \int_{U}^r \frac{du}{\sqrt{\theta\log(1+u^2)-u^2}},
$$
which implies that the first and third integrals in the decomposition formula cancel out.
Since the total integral does not depend on $r$, its value is computed from the second integral in the limit $r\to 0$:
\begin{eqnarray*}
& \phantom{t} & \int_{u^*}^{-U}\frac{du}{\sqrt{\theta\log(1+u^2)-u^2}} - \int_{u^*}^{U}\frac{du}{\sqrt{\theta\log(1+u^2)-u^2}} \\
& = & \int_{C_r} \frac{du}{\sqrt{\theta\log(1+u^2)-u^2}} \to \frac{\pi i}{\sqrt{\theta - 1}} \quad \mbox{\rm as} \quad r \to 0.
\end{eqnarray*}
This implies formula (\ref{Symmetry_U}). Applying the symmetry property (\ref{z-symmetry})
we obtain the same formula (\ref{Symmetry_U}) for $U\in (0,u^*)$ and the lower edge of the cut $\mathcal{Q}_3$. $\Box$ \\

\noindent {\em Proof of (c).} The function $z(U)$ maps the point $U=i$ into
\begin{eqnarray}
z_0 = \int_{\gamma_0} \frac{du}{\sqrt{\theta\log(1+u^2)-u^2}},
\label{z_0}
\end{eqnarray}
where $\gamma_0$ is a path that connects the points $u^*$ and $i$,
and lies in $Q$. We take the path $\gamma_0$ in Fig.\ref{Path}(c) as a union of interval of real axis $[r,u^*]$,
arc $C_r$ of the circle of radius $r$, and the interval on imaginary axis $[ir, i]$,
where $r$ can be taken arbitrarily small. For this choice of $\gamma_0$ one has
\begin{eqnarray}
z_0=I_r(r)+I_C(r)+I_i(r),\label{z_0_sum}
\end{eqnarray}
where
\small
\begin{eqnarray*}
I_r(r)= \int_{u^*}^r\frac{du}{\sqrt{\theta\log(1+u^2)-u^2}},\quad
I_C(r) = \int_{C_r}\frac{du}{\sqrt{\theta\log(1+u^2)-u^2}}, \quad
I_i(r) = \int_{ir}^i\frac{du}{\sqrt{\theta\log(1+u^2)-u^2}}.
\end{eqnarray*}
\normalsize
The value of $z_0$ does not depend on $r$, whereas each of summands in (\ref{z_0_sum}) does.

Consider the limit $r\to0$. Both the integrals $I_r(r)$ and $I_i(r)$ diverge as $r \to 0$. However, let us show that the sum $I_r(r)+I_i(r)$ has a finite limit as $r \to 0$. By means of parametrization $u=i\xi$, integral $I_i(r)$ can be rewritten in the form
\begin{eqnarray*}
I_i(r)=\int_{r}^1\frac{d\xi}{\sqrt{-\theta\log(1-\xi^2)-\xi^2}},
\end{eqnarray*}
therefore, both $I_r(r)$ and $I_i(r)$ are real. Summing up $I_r(r)$ and $I_i(r)$ yields
\small
\begin{eqnarray*}
I_r(r) + I_i(r) & = & -\int_{r}^{u^*}\frac{du}{\sqrt{\theta\log(1+u^2)-u^2}}+
\int_r^1\frac{du}{\sqrt{-\theta\log(1-u^2)-u^2}}=\\[2mm]
& = & -\int_{1}^{u^*}\frac{du}{\sqrt{\theta\log(1+u^2)-u^2}}+\int_r^1\frac{\sqrt{\theta\log(1+u^2)-u^2}-\sqrt{-\theta\log(1-u^2)-u^2}}
{\sqrt{\theta\log(1+u^2)-u^2}\sqrt{-\theta\log(1-u^2)-u^2}}.
\end{eqnarray*}
\normalsize
By multiplying the numerator and denominator of the last integrand by
\begin{eqnarray*}
\sqrt{\theta\log(1+u^2)-u^2}+\sqrt{-\theta\log(1-u^2)-u^2}
\end{eqnarray*}
we conclude that the last integral converges when $r\to 0$.
Passing to the  limit $r\to 0$ yields the real-valued coefficient
\begin{eqnarray}
\label{alpha}
\alpha := -\lim_{r \to 0} \left[ I_r(r) + I_i(r) \right] = J_1 + J_2,
\end{eqnarray}
where $J_1$ and $J_2$ are defined below (\ref{alpha-beta}).

Consider now the integral $I_C(r)$. In the limit $r\to 0$
the logarithm can be replaced by its Taylor expansion and the integral can be calculated explicitly
\begin{eqnarray}
\label{beta}
i \beta := \lim_{r \to 0} I_C(r) = \frac{\pi i}{2\sqrt{\theta-1}}.
\end{eqnarray}
Limits (\ref{alpha}) and (\ref{beta}) recover the expressions (\ref{alpha-beta}) for $\alpha, \beta > 0$. So, the point $U=i$ maps to $z_0=-\alpha+i\beta$ and due to (\ref{z-symmetry}) the point $U=-i$ maps to $-\overline{z}_0=\alpha+i\beta$. $\Box$ \\

\begin{figure}[h]
\centerline{\includegraphics [scale=0.75]{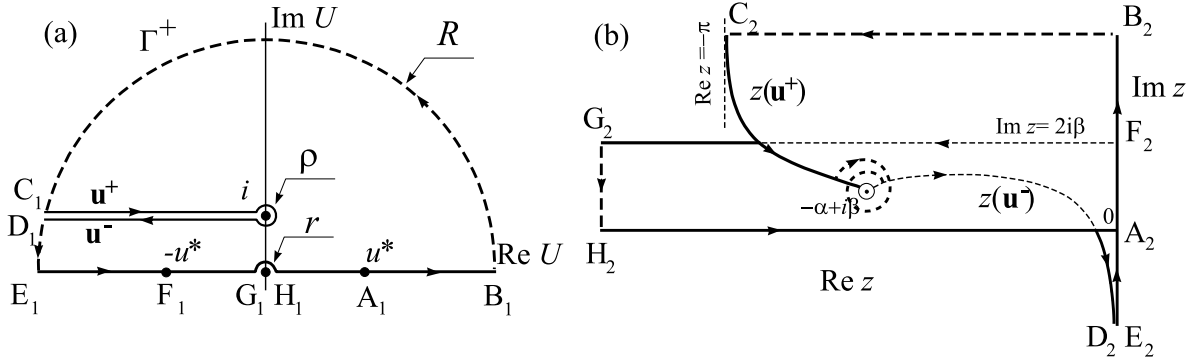}}
\caption{The contour $\Gamma_+$ (a) and its image on $z$-plane (b).
The points $A_2,\ldots,H_2$ are the images of the points $A_1,\ldots,H_1$. } \label{Plane_Z_map}
\end{figure}

\noindent {\em Proof of (d).}  Consider the upper part of the set $Q$ situated in upper half-plane.
Introduce the contour $\Gamma^+$ in Fig.~ \ref{Plane_Z_map}(a). It passes along the big circle
(arcs $B_1C_1$ and $D_1E_1$) that is  centered in the origin and has a large enough radius $R$,
includes the intervals of the real axis $E_1G_1$ and $H_1B_1$ and the semi-circle of a small radius $r$
that avoids the pole in the origin (the arc $G_1H_1$). The contour $\Gamma^+$ also includes the path getting
round the cut ${\mathcal Q}_1$ that consists of the line segments ${\bf u}^+$, ${\bf u}^-$, and the circle $C_i^\rho$.

Let us analyse the image of the contour $\Gamma^+$ in Fig.~\ref{Plane_Z_map} (b).
Consider the points $A_1$, $B_1$, $E_1$, $G_1$ $H_1$ on the real axis in the $U$-plane.
Evidently, $z(u^*)=0$, so the point $A_1$ maps into the origin
in the $z$-plane (the point $A_2$). It follows directly from formula (\ref{exact-complex}) that the
interval $A_1H_1$ situated on the real axis  maps into the interval $A_2H_2$ that lies on the real
negative semi-axis in $z$-plane. Next, since $f'(u^*)<0$, it is straightforward to check that $A_1B_1$
maps to the interval $A_2B_2$ of the positive imaginary semi-axis. According to formulas
(\ref{Small_r01})-(\ref{Small_r02}), the arc $H_1G_1$ of small semi-circle  maps into a distant
curve segment $H_2G_2$. The smaller is the radius $r$ of the semi-circle, the greater is the
distance of $H_2G_2$ from the origin in the $z$-plane. Due to (\ref{Symmetry_U}), the imaginary
part of $G_2$ is equal to $2\pi\beta$ and its real part tend to $-\infty$ as $r$ tends to zero.
Also due to (\ref{Symmetry_U}), $z(-u^*)=2i\beta$. Finally, the image of interval $E_1F_1$
lies on the imaginary axis in the $z$-plane.

Consider the great semi-circle (arcs $B_1C_1$ and $D_1E_1$). It follows directly from formula (\ref{exact-complex})
that when $R$ tends to infinity the image of the arc $B_1C_1$ tends to a distant line segment of length $\pi$
that is parallel to the real axis. Similarly, when $R$ tends to infinity the images of the points $D_1$ and $E_1$
tend to each other and their imaginary parts tend to infinity.

At last, consider the images of the line segments ${\bf u}^+$, ${\bf u}^-$, and the circle $C_i^\rho$.
As it was shown in (c), the point $i$ maps into $z=-\alpha+i\beta$ where $\alpha$ and $\beta$
are given by formulas (\ref{alpha-beta}). Let $u(t)=i-t$, $t>\rho$ be the parametrization
on ${\bf u}^\pm$. The images of ${\bf u}^\pm$  are given by
\begin{eqnarray*}
z^-(t)&=&-\alpha+i\beta+\int_i^{i-t}\frac{du}{\sqrt{\theta\log(1+u^2)-u^2}}\\
z^+(t)&=&-\alpha+i\beta+\int_i^{i-t}\frac{du}{\sqrt{\theta(\log(1+u^2)+2\pi i)-u^2}}
\end{eqnarray*}
where $t$ decreases at $z^+(t)$ and increases at $z^-(t)$. Then
\begin{eqnarray*}
\frac{d z^-}{dt}&=-&\frac1{\sqrt{\theta\log(1+(i-t)^2)-(i-t)^2}}\\
\frac{d z^+}{dt}&=-&\frac1{\sqrt{\theta(\log(1+(i-t)^2)+2\pi i)-(i-t)^2}}
\end{eqnarray*}
The behaviour of the functions
\begin{eqnarray*}
f({\bf u}^-) = \theta\log(1+(i-t)^2)-(i-t)^2,\quad f({\bf u}^+) = \theta(\log(1+(i-t)^2)+2\pi i)-(i-t)^2
\end{eqnarray*}
is described in Lemma \ref{lemma-00}, from which it follows that
\begin{eqnarray*}
\RE~[\sqrt{f({\bf u}^-)}]<0,\quad \RE~[\sqrt{f({\bf u}^+)}]>0,\quad t>\rho.
\end{eqnarray*}
When $U$ moves along ${\bf u}^-$ and ${\bf u}^+$ in  directions indicated by arrows on Fig.~\ref{Plane_Z_map}(a),
the corresponding point $\RE~[z(U)]$ on Fig.~\ref{Plane_Z_map}(b) increases in both cases.
This implies that the image of the area inside $\Gamma^+$ is multi-sheeted and covers completely the half-strip
$\{~\RE~z\leq0, \;\; 0\leq\IM~z\leq 2\beta\}$ with the cut $\mathcal{S}_1=[-\alpha+i\beta,-\alpha+2i\beta]$.

Passing to the limits $R\to\infty$, $\rho\to 0$ and $r\to 0$ and employing the symmetry property (\ref{z-symmetry}) we conclude that
the set $S$ belongs to the image of $Q$.
\end{proof}

Theorem \ref{theorem-001} implies the following corollary, which is important for further applications.

\begin{corollary}
\label{theorem-000}
Let $\theta>1$ and $\alpha,\beta$ are given by formulas (\ref{alpha-beta}).
The solitary wave solution $U(x)$ can be analytically continued to $S\subset\mathbb{C}$  where $S$  is
the strip $\{0\leq\IM~z\leq 2\beta\}$ with two vertical cuts $\mathcal{S}_1=[-\alpha+i\beta,-\alpha+2i\beta]$
and $\mathcal{S}_2=[\alpha+i\beta,\alpha+2i\beta]$ shown on Fig.~\ref{Tilde_Q}(b).
The resulting function $U(z)$ is single-valued in the interior of $S$ and
\begin{itemize}
\item[(a)] if $0\leq y \leq 2\beta$ then $\lim_{R\to\pm\infty} U(R+iy)=0$;
\item[(b)] if $x<-\alpha$ or $x>\alpha$ then $U(x+2i\beta)=-U(x)$.
\end{itemize}
\end{corollary}

\begin{proof}
The function $U(z)$ defined by implicit formula (\ref{exact-complex}) coincides with $U(x)$
on the real axis. By Theorem \ref{theorem-001}, the function $U(z)$ is defined in $S$.
This implies that $U(z)$ is an analytic continuation of $U(x)$ to $S$. Let $U(S)\subset Q$
be the image of $S$ on the $U$-plane. Note, that if $\tilde{z}$ is an arbitrary internal point of $S$
and $\tilde{U}=U(\tilde{z})$, then $z'(\tilde{U}) \ne 0$ and $U'(\tilde{z}) \ne 0$.
This means that there is one-to one-correspondence between some neighbourhood
of $\tilde{z}$ on the $z$-plane and some neighbourhood of $\tilde{U}$ on the $U$-plane. Therefore, (i) there are no
two different internal points $z_1,z_2 \in S$ such that $U(z_1)=U(z_2)$ and
(ii) there are no two different points $U_1,U_2$ in the interior of $U(S)$ such that $z(U_1)=z(U_2)$.
Hence $U(z)$ is a single-valued function in the interior of $S$ and $z(U)$ is a single-valued function in the interior of $U(S)$.
The assertion (a) follows from the formulas (\ref{Small_r01})-(\ref{Small_r02}).
The assertion (b) follows from (\ref{Symmetry_U}).
\end{proof}

\subsection{Asymptotic properties of $U(z)$}\label{AsPropU}

The local behavior of the solution $U(z)$ near the singularity $z_0$
with ${\rm Re}(z_0) < 0$ and ${\rm Im}(z_0) > 0$ is prescribed by the following result.

\begin{lemma}
For every $z$ near $z_0$ with $\arg(z_0 - z) \in \left(-\frac{\pi}{2},\frac{3\pi}{2}\right)$,
the solution $U$ satisfies
\begin{equation}
\label{singular-behavior}
U(z) = i + \sqrt{\theta} (z - z_0) \sqrt{\log(z_0-z)} \left[ 1 +
\mathcal{O}\left(\frac{\log|\log|z-z_0||}{\log|z-z_0|}\right) \right] \quad \mbox{\rm as} \quad z \to z_0,
\end{equation}
where $\sqrt{u}$ is defined at the main branch with $\arg(u) \in (0,2\pi)$.
\label{lemma-4}
\end{lemma}

\begin{proof}
Combining (\ref{exact-complex}) and (\ref{z_0}) yields the formula
$$
z - z_0 = \int_{i}^{U} \frac{du}{\sqrt{\theta\log(1+u^2)-u^2}}.
$$
Let us first define $U$ on the imaginary axis below $i$ so that we can write
$U = i(1-V)$ with $V$ real and positive. Using the similar representation
for the integration variable $u = i (1-v)$ yields
$$
z - z_0 = (-i) \int_{0}^V \frac{dv}{\sqrt{\theta \log v + \theta \log(2 - v) + (1-v)^2}}.
$$
In the limit $V \to 0$, the integrand can be expanded as
$$
\int_{0}^V \frac{dv}{\sqrt{\theta\log v}} \left[ 1 + \mathcal{O}\left(\frac{1}{|\log v|} \right) \right].
$$
Since
$$
\frac{1}{\sqrt{\log v}} = \frac{d}{dv} \left[ \frac{v}{\sqrt{\log v}} \right] + \frac{1}{2 \sqrt{(\log v)^3}},
$$
the integral is represented asymptotically as
\begin{eqnarray}
z - z_0 = \frac{-i V}{\sqrt{\theta \log V}} \left[ 1 + \mathcal{O}\left(\frac{1}{|\log V|}\right) \right] \quad \mbox{\rm as} \quad V \to 0. \label{Label}
\end{eqnarray}
Define $\sqrt{u}$ at the main branch with $\arg(u) \in (0,2\pi)$ so that
if $V$ is real and positive, then $z - z_0$ is real and negative.
If $\arg(V) \in \left(-\frac{\pi}{2},\frac{3\pi}{2}\right)$ so that
$\arg(U - i) \in (-\pi,\pi)$ like on Fig. \ref{Tilde_Q}(a),
then $\arg(z-z_0) \in \left(-\frac{3\pi}{2},\frac{\pi}{2} \right)$
like on Fig. \ref{Tilde_Q}(b). Hence, the function
(\ref{Label}) is continued in the open region with $\arg(z_0 - z) \in \left(-\frac{\pi}{2},\frac{3\pi}{2}\right)$.

It remains to justify the asymptotic expansion (\ref{singular-behavior}).
To do so, we use the implicit function theorem. By substitution
\begin{equation}
\label{V-representation}
V = i \sqrt{\theta \log(z_0-z)} (z-z_0) W,
\end{equation}
we convert the expansion (\ref{Label}) to the nonlinear equation:
\begin{eqnarray*}
W \left[ 1 + \mathcal{O}\left(\frac{1}{|\log W| + |\log(z_0-z)|}\right) \right]
= \sqrt{1 + \frac{\log W}{\log(z_0-z)} + \frac{\log(\log(z_0-z)) + \log \theta  - \pi}{2 \log(z_0-z)}}.
\end{eqnarray*}
Let us define
\begin{eqnarray*}
\mu := \frac{\log(\log(z_0-z))}{\log(z_0-z)}, \quad \nu := \frac{1}{\log(z_0-z)},
\end{eqnarray*}
so that $\mu \to 0$ and $\nu \to 0$ as $z \to z_0$ along any path in the domain on Fig. \ref{Tilde_Q}(b).
Since $\mu = -\nu \log(\nu)$, the two variables are dependent of each other and $|\nu| \ll |\mu|$.
Fix the path $z \to z_0$ and invert the map $\mathbb{C} \ni \nu \mapsto \mu := -\nu \log(\nu) \in \mathbb{C}$
to obtain the map $\mu \mapsto \nu$ satisfying $\lim_{\mu \to 0} \nu(\mu) = \lim_{\mu \to 0} \nu'(\mu) = 0$.
The nonlinear equation for $W$ can then be rewritten as the root-finding problem $F(W,\mu) = 0$,
where
\begin{eqnarray*}
F(W,\mu) :=
W \left[ 1 + \mathcal{O}\left(\frac{|\nu(\mu)|}{1 + |\nu(\mu)| |\log W|}\right) \right]
- \sqrt{1 + \frac{1}{2} \mu + \frac{1}{2} (\log \theta  - \pi + \log W) \nu(\mu)}.
\end{eqnarray*}
The function $F(W,\mu) : \mathbb{C} \times \mathbb{C} \to \mathbb{C}$ is $C^1$ in $W$ at $W = 1$
and $C^1$ in $\mu$ along the path $\mu \to 0$ with $\lim_{\mu \to 0} F(1,\mu) = 0$,
$\lim_{\mu \to 0} \partial_W F(1,\mu) = 1$, and $\lim_{\mu \to 0} \partial_{\mu} F(1,\mu) = -1$. By the implicit function theorem,
there is an unique $C^1$ map $\mu \mapsto W$ along the path $\mu \to 0$ such that $\lim_{\mu \to 0} W(\mu) = 1$ and
$\lim_{\mu \to 0} W'(\mu) = 1$, which is written in the original variables as follows:
$$
W(z) = 1 + \mathcal{O}\left(\frac{\log|\log|z-z_0||}{\log|z-z_0|} \right).
$$
Substitution of this expansion to $U = i (1 - V)$ with $V$ given by (\ref{V-representation})
yields expansion (\ref{singular-behavior}) for every $z$ near $z_0$ with
$\arg(z_0 - z) \in \left(-\frac{\pi}{2},\frac{3\pi}{2}\right)$,
\end{proof}

The symmetry reflection (\ref{Symm}) yields the local behaviour of the solution near the symmetric
singularity $z_0^* = -\bar{z}_0$ with ${\rm Re}(z_0^*) > 0$.

\begin{corollary}
For every $z$ close to $z_0^* = -\bar{z}_0$ with $\arg(z-z_0^*) \in \left(-\frac{3\pi}{2},\frac{\pi}{2}\right)$,
the solution $U$ satisfies $U(z) = \overline{U(-\bar{z})}$ with
\begin{equation}
\label{singular-behavior-2}
U(z) = -i - \sqrt{\theta} (z - z_0^*) \sqrt{\log(z_0^* - z)} \left[ 1 +
\mathcal{O}\left(\frac{\log|\log|z-z_0^*||}{\log|z-z_0^*|}\right) \right] \quad
\mbox{\rm as} \quad  z \to z_0^*.
\end{equation}
\label{corollary-4}
\end{corollary}


Finally, we define the Fourier transform of the solitary wave solution $U$ by
\begin{equation}
\label{FT-of-U}
I(\varkappa) := \int_{\mathbb{R}} U(x) e^{i \varkappa x} dx, \quad \varkappa \in \mathbb{R}.
\end{equation}
The following lemma computes the asymptotic behavior of the Fourier integral
$I(\varkappa)$ as $\varkappa \to \infty$.

\begin{lemma}
\label{lemma-5}
It is true that
\begin{gather}\label{I_varkappa}
I(\varkappa) = \frac{2\pi\sqrt{\theta}}{\varkappa^2\sqrt{\log \varkappa}}e^{-\beta\varkappa}\cos(\alpha\varkappa)
\left[ 1 + \mathcal{O}\left(\frac{1}{\log\varkappa}\right) \right],\quad \varkappa\to\infty.
\end{gather}
\end{lemma}

\begin{proof}
Consider the contour $\Gamma_1=A_1B_1C_1D_1E_1F_1G_1$ shown in Fig.\ref{Contour} and the integral
\begin{gather*}
\int_{\Gamma_1} U(z) e^{i \varkappa z} dz.
\end{gather*}
By Theorem \ref{theorem-001}, the integrand is analytic inside $\Gamma_1$, hence the integral is equal to zero.
Therefore, we decompose the integral into the sum of integrals
\begin{align}
0 = \int_{\Gamma_1} U(z) e^{i \varkappa z} dz&=\int_{-R}^R+\int_{[E_1F_1]}+\int_{[F_1G_1]}+\int_{{\it   l}_\downarrow^+}+
\int_{C_\rho^+}+\int_{{\it l}_\uparrow^+}+\nonumber\\[2mm]
&+\int_{\alpha+2i\beta}^{-\alpha+2i\beta}+
\int_{{\it l}_\downarrow^-}+
\int_{C_\rho^-}+\int_{{\it l}_\uparrow^-}+\int_{[A_1B_1]}+\int_{[B_1C_1]}U(z) e^{i \varkappa z}~dz
\label{SumInt}
\end{align}
and consider each integral consecutively. We have
\begin{equation}
\label{relat-1}
\lim_{R \to \infty} \int_{-R}^R U(z) e^{i \varkappa z} dz = I(\varkappa).
\end{equation}
Thanks to (a) in Corollary \ref{theorem-000}, we obtain
\begin{equation}
\label{relat-2}
\lim_{R \to \infty} \int_{[E_1F_1]} U(z) e^{i \varkappa z} dz
+ \int_{[B_1C_1]} U(z) e^{i \varkappa z} dz = 0.
\end{equation}
By using the parametrization $z=t+2i\beta$ and the symmetry property in (b) of Corollary \ref{theorem-000},
we obtain
\begin{align}
& \phantom{t} \lim_{R \to \infty} \left|\int_{[A_1B_1]} U(z) e^{i \varkappa z}~dz\right|+\left|\int_{[F_1G_1]} U(z) e^{i \varkappa z}~dz\right| \nonumber \\[2mm]
& =  \left|\int_{-\infty}^{-\alpha} U(t+2i\beta)e^{-2\beta\varkappa} e^{it\varkappa}~dt\right|+\left|\int_{\alpha}^{\infty} U(t+2i\beta)e^{-2\beta\varkappa} e^{it\varkappa}~dt\right| \nonumber\\[2mm]
&\leq \int_{-\infty}^{-\alpha}\left|U(t)\right|e^{-2\beta\varkappa} ~dt+\int_{\alpha}^{\infty} \left|U(t)\right|e^{-2\beta\varkappa} ~dt
\leq e^{-2\beta\varkappa}\int_{\mathbb{R}} \left|U(t)\right| dt.
\label{UpSidesEst}
\end{align}
Because the function $U$ is bounded on the interval $[-\alpha+2i\beta;\alpha+2i\beta]$, we obtain
\begin{gather}\label{Central}
\left|\int_{\alpha+2i\beta}^{-\alpha+2i\beta} U(z)e^{i\varkappa z}~dz\right|\leq 2\beta \max_{[-\alpha+2i\beta;\alpha+2i\beta]}|U(z)|e^{-2\beta\varkappa}.
\end{gather}
It remains to estimate the integrals
\begin{gather*}
I_+(\varkappa)=\int_{{\it   l}_\downarrow^+}+
\int_{C_\rho^+}+\int_{{\it l}_\uparrow^+} U(z)e^{i\varkappa z}~dz,
\quad
I_-(\varkappa)=\int_{{\it   l}_\downarrow^-}+
\int_{C_\rho^-}+\int_{{\it l}_\uparrow^-} U(z)e^{i\varkappa z}~dz.
\end{gather*}

\begin{figure}
{\centerline{\includegraphics [scale=0.6]{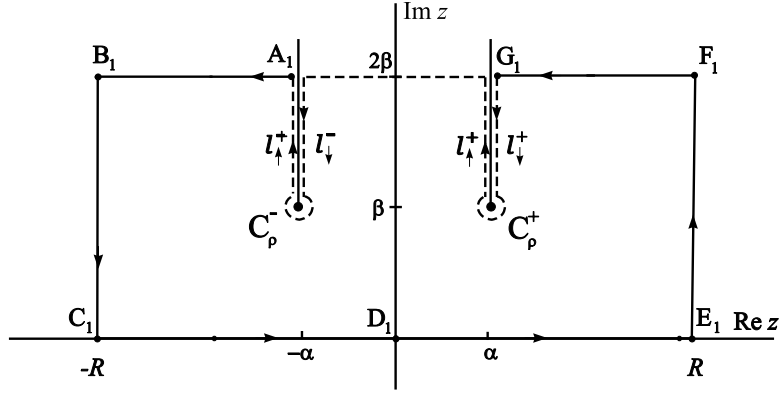}}}
\caption{The contour $\Gamma_1$  for the proof of (\ref{I_varkappa}).} \label{Contour}
\end{figure}

Formulas (\ref{SumInt}), (\ref{relat-1}), (\ref{relat-2}), (\ref{UpSidesEst}), and (\ref{Central}) imply that
\begin{equation}
\label{I-representation}
I(\varkappa)=-I_+(\varkappa)-I_-(\varkappa)+\mathcal{O}\left(e^{-2\beta\varkappa}\right),
\end{equation}
hence we need to determine the asymptotical behavior of $I_{\pm}(\varkappa)$ as $\varkappa\to\infty$.
Thanks to the singular behavior (\ref{singular-behavior-2}), the integral
\begin{gather*}
\int_{C_\rho^+}U(z)e^{i\varkappa z}~dz
\end{gather*}
tends to zero as $\rho\to 0$. Therefore
\begin{gather*}
I_+(\varkappa)=\int_{\alpha+i\beta}^{\alpha+2i\beta} U_1(z)e^{i\varkappa z}~dz+\int_{\alpha+2i\beta}^{\alpha+i\beta} U_2(z)e^{i\varkappa z}~dz,
\end{gather*}
where $U_1(z)$ and $U_2(z)$ are the values of $U(z)$ on both sides of the branch cut at ${\rm Re}(z) = \alpha$ and
${\rm Im}(z) > \beta$. Introducing parametrization $z=\alpha+i(\beta+t)$ on the path of the integration  one has
\begin{align*}
I_+(\varkappa) = i e^{(-\beta+i\alpha)\varkappa} \int_{0}^{\beta}
\left[ U_1(\alpha+i(\beta+t))-U_2(\alpha+i(\beta+t)) \right] e^{-t\varkappa}~dt,
\end{align*}
where the boundary values satisfy the singular behavior from (\ref{singular-behavior-2}):
\begin{align}
U_1(\alpha+i(\beta+t))&=-i-i\sqrt{\theta}t\sqrt{\log t + \frac{\pi i}2} \left[ 1 +
\mathcal{O}\left(\frac{\log|\log t|}{|\log t|}\right) \right],\quad t\to +0, \label{As_U1} \\[2mm]
U_2(\alpha+i(\beta+t))&=-i-i\sqrt{\theta}t\sqrt{\log t + \frac{5\pi i}2} \left[ 1 +
\mathcal{O}\left(\frac{\log|\log t|}{|\log t|}\right) \right],\quad t\to +0. \label{As_U2}
\end{align}
Therefore, we obtain as $t\to +0$:
\begin{eqnarray}
\nonumber \tilde U(t) & := & U_1(\alpha+i(\beta+t))-U_2(\alpha+i(\beta+t)) \\
& = & -i\sqrt{\theta}t\left(\sqrt{\log t + \frac{\pi i}2}-\sqrt{\log t + \frac{5\pi i}2}\right)\left[ 1 +
\mathcal{O}\left(\frac{\log|\log t|}{|\log t|}\right) \right] \nonumber\\
&=& - \frac{2\pi\sqrt{\theta} t}{\sqrt{\log t + \frac{\pi i}2}+\sqrt{\log t + \frac{5\pi i}2}}\left[ 1 +
\mathcal{O}\left(\frac{\log|\log t||}{|\log t|}\right)\right] \nonumber \\
&=& \frac{i\pi\sqrt{\theta} t}{\sqrt{|\log t|}}\left[ 1 +
\mathcal{O}\left(\frac{\log|\log t|}{\sqrt{\log|t|}}\right)\right].
\label{I-integrand}
\end{eqnarray}
The integral $I_+(\varkappa)$ computed at the integrand (\ref{I-integrand})
is the Laplace integral with logarithmic singularity at $t=0$.
The asymptotical behavior of $I_+(\varkappa)$ as $\varkappa\to\infty$ is found by the Laplace method,
see formula (1.38) on p. 48 in \cite{Fedoryuk},
\begin{gather}\label{FedForm}
\int_0^{t_0}t^{b-1}|\log t|^ce^{-\varkappa t} f(t)~dt\sim\varkappa^{-b}(\log \varkappa)^c
\sum_{k=0}^\infty a_k \; (\log \varkappa)^{-k},\quad
a_0 = \Gamma(b) f(0),
\end{gather}
where $f(t)\in C^1[0,t_0]$, $b>0$ and $c\in \mathbb{R}$. Making use of the asymptotic
formula (\ref{FedForm}) with $b=2$ and $c=-1/2$ yields
\begin{gather*}
I_+(\varkappa)=-\frac{\pi\sqrt{\theta}}{\varkappa^2\sqrt{\log \varkappa}}e^{(-\beta+i\alpha)\varkappa}
\left[1 + \mathcal{O}\left(\frac{1}{\log\varkappa}\right) \right],\quad \varkappa\to\infty.
\end{gather*}
In the same way, we obtain
\begin{gather*}
I_-(\varkappa)=-\frac{\pi\sqrt{\theta}}{\varkappa^2\sqrt{\log \varkappa}}
e^{(-\beta-i\alpha)\varkappa} \left[ 1 + \mathcal{O}\left(\frac{1}{\log\varkappa}\right) \right],\quad \varkappa\to\infty.
\end{gather*}
By using (\ref{I-representation}) and neglecting the smaller exponential terms, we finally obtain (\ref{I_varkappa}).
\end{proof}

\section{Solitary wave solution to the fourth-order equation}

Here we consider the fourth-order differential equation
\begin{eqnarray}
\varepsilon^2 \frac{d^4 u}{d x^4} + \frac{d^2 u}{dx^2} + u -\frac{\theta u}{1+u^2}=0,
\label{ode}
\end{eqnarray}
where $\varepsilon$ is a small positive parameter. Equation (\ref{ode}) arises as the next-order continuous
approximation for the advance-delay equation (\ref{advance-delay-intro}) taking into account the expansion
(\ref{expansion-difference}) with the correspondence
\begin{equation}
\varepsilon = \frac{h}{2 \sqrt{3}}.\label{h-eps}
\end{equation}
The main goal of this section is
to describe a countable sequence of solitary wave solutions with $\varepsilon$ near $\{ \varepsilon_m \}_{m \in \mathbb{N}}$,
where the sequence $\{ \varepsilon_m \}_{m \in \mathbb{N}}$
accumulates to zero as $m \to \infty$ according to the asymptotic representation:
\begin{equation}
\label{eps-ode}
\eps_m \sim \frac{2 \alpha}{\pi (2m-1)}, \quad m \in \mathbb{N}.
\end{equation}
where $\alpha > 0$ is defined by (\ref{alpha-beta}).
In particular, the spacing between two consequent values of the sequence is asymptotically given by
\begin{equation}
\label{spacing-ode}
\frac{1}{\varepsilon_{m+1}} - \frac{1}{\varepsilon_m} \to \frac{\pi}{\alpha} \quad \mbox{\rm as} \quad m \to \infty.
\end{equation}
With the correspondence (\ref{h-eps}), the asymptotic formula (\ref{eps-ode}) is equivalent to (\ref{spacing-odeintro}).

We obtain the asymptotic values (\ref{eps-ode}) by means of two analytical methods, one relies on the semi-classical analysis
of oscillatory integrals (Section \ref{sec-3-1}) and the other one relies on the beyond-all-order asymptotic expansions
(Section \ref{sec-3-2}). Neither method is rigorous and has been fully justified. Nevertheless, the outcomes of the two methods
are identical and these outcomes are confirmed by the numerical results (Section \ref{NumResults}).

\subsection{Analysis of oscillatory integrals}
\label{sec-3-1}

Let $U$ be the even, positive, and exponentially decaying solution
to the second-order equation (\ref{e=0}) defined in the implicit form by (\ref{exact}).
We are looking for an even solution to the fourth-order equation (\ref{ode}) in the perturbed form $u = U + v$.
Substitution yields the following persistence problem for $v$:
\begin{equation}
\label{perturbed-problem}
L_{\eps} v = H_{\eps} + N(v),
\end{equation}
where
\begin{equation}
\label{operator-L}
L_{\eps} := -\varepsilon^2 \frac{d^4}{d x} - \frac{d^2}{dx^2} + \theta - 1
\end{equation}
is the linearization operator at the zero solution,
\begin{equation}
\label{source-term-H}
H_{\eps} := \eps^2 \frac{d^4 U}{d x^4}
\end{equation}
is the source term, and
\begin{equation}
\label{nonlinear-term-N}
N(v) = -\theta \frac{U^2 (3 + U^2)}{(1+U^2)^2} v +
\theta v^2 \frac{U(3-U^2) + v (1-U^2)}{(1+U^2)^2 (1 + U^2 + 2 U v + v^2)}
\end{equation}
include both linear and nonlinear terms in $v$. If the source term $H_{\eps}$ is zero (if $\eps = 0$),
there exists a solution $v = 0$, hence one can hope that
small $H_{\eps}$ for small $\eps \neq 0$ generates small $v$ in ${\rm Dom}(L_{\eps}) = H^4(\mathbb{R})$
satisfying equation (\ref{perturbed-problem}).
Unfortunately, $L_{\eps}$ is not a Fredholm operator in $L^2(\mathbb{R})$ because $0 \in \sigma(L_{\eps})$.
Since $\sigma(L_{\eps})$ is purely continuous, a bounded solution $v$ of the inhomogeneous equation
\begin{equation}
\label{linear-problem}
L_{\eps} v = H_{\eps}
\end{equation}
in the space of even functions with $v(-x) = v(x)$ for $x \in \mathbb{R}$
develops generally oscillations in $x$ as $|x| \to \infty$ \cite{Volpert,Volpert2}.
The only possibility to avoid oscillations in the bounded solution $v$
solving the inhomogeneous equation (\ref{linear-problem}) is to satisfy the constraint
$I_{\eps} = 0$, where
\begin{equation}
\label{integral-I}
I_{\eps} := \int_{\mathbb{R}} H_{\eps}(x) e^{i k_{\eps} x} dx.
\end{equation}
Here $k_{\eps}$ is the only real positive root of $D_{\eps}(k) = 0$, where
$$
D_{\eps}(k) := -\eps^2 k^4 + k^2 + \theta - 1, \quad k \in \mathbb{R}
$$
is the dispersion relation for the operator $L_{\eps}$.
It is clear that $k_{\eps} = \eps^{-1} + \mathcal{O}(1)$ as $\eps \to 0$
and in particular, $k_{\eps} \to \infty$ as $\eps \to 0$. As is shown in \cite{Volpert},
if $I_{\eps} = 0$, then $v = L_{\eps}^{-1} H_{\eps} \in H^4(\mathbb{R})$.
As is argued heuristically in \cite{prl-alfimov}, if $I_{\eps_0} = 0$
for some small $\eps_0$, then there exists a unique solution $v \in H^4(\mathbb{R})$
to the persistence problem (\ref{perturbed-problem}) for $\eps$ near $\eps_0$.

Hence, we are looking for zeros of $I_{\eps}$ as $\eps \to 0$.
Integrating (\ref{integral-I}) by parts four times yields the equivalent expression for $I_{\eps}$:
\begin{equation}
\label{I-aux}
I_{\eps} = k^4_{\eps}\eps^2\int_{\mathbb{R}} U(x) e^{i k_{\eps} x} dx \equiv k^4_{\eps}\eps^2 I(k_{\eps}),
\end{equation}
where $I(\varkappa)$ with $\varkappa = k_{\eps}$ is given by (\ref{FT-of-U}).
Since $k_{\eps} = \eps^{-1} + \mathcal{O}(1)$ as $\eps \to 0$,
substituting the asymptotic behaviour (\ref{I_varkappa}) into (\ref{I-aux}) yields the asymptotic behavior
\begin{eqnarray}
I_{\eps} \sim \frac{2\pi\sqrt{\theta}}{\sqrt{\log (1/\eps)}}e^{-\beta/\eps}\cos(\alpha/\eps) \quad
\mbox{\rm as} \quad \eps \to 0.
\label{main-ode}
\end{eqnarray}
The leading order of $I_{\eps}$ vanishes at $\{ \eps_m \}_{m \in \mathbb{N}}$
given by (\ref{eps-ode}).

\subsection{Beyond-all-order asymptotics}
\label{sec-3-2}

By studying the fourth-order equation (\ref{ode}) using beyond-all-order methods,
we can recover the asymptotic result \eqref{eps-ode}. In particular, we will show that
the asymptotic solution contains two Stokes lines, each of which switches on an exponentially small
contribution which does not decay in the far field. Solitary wave solutions are associated with
the special cases in which the two contributions cancel.

The central idea of exponential asymptotics is that a divergent asymptotic series expansion
can be truncated optimally, and when this occurs, the truncation remainder is exponentially
small in the small parameter \cite{Berry0,Berry1,Boyd2}. By rescaling the problem to obtain
an equation for the remainder, it is possible to isolate exponentially small contributions
to the asymptotic solution behaviour, which are typically invisible to classical asymptotic power series methods.

The process we use for identifying Stokes lines is based on the matched asymptotic expansion technique described in \cite{Daalhuis1}. We incorporate the use of late-order term analysis, devised by \cite{Chapman1}, which extends the matched asymptotic expansion technique so that it may be applied to nonlinear differential equations. The steps of this method are as follows:
\begin{itemize}
\item Determine the behaviour of the \textit{late-order} asymptotic terms of the solution; that is, expand the solution as an asymptotic power series in a small parameter and then obtain an asymptotic approximation for the $j$th series term in the limit that $j \rightarrow \infty$.
\item Use the asymptotic form of the late-order terms to optimally truncate the asymptotic series. Rescale the equation to obtain an expression for the remainder term.
\item Perform a local asymptotic analysis of the remainder term in the neighbourhood of Stokes lines, and apply matched asymptotic expansions in order to determine the exponentially small quantity that is switched on as the Stokes line is crossed.
\end{itemize}

Using this method, we will establish that the exponentially small oscillations present in the solution $u(x)$, denoted by $u_{\mathrm{osc}}(x)$,
have the following asymptotic behaviour as $\epsilon \rightarrow 0$:
\begin{equation}
u_{\mathrm{osc}}(x) \sim \left\{
        \begin{array}{ll}
            0, & \quad x < -\alpha - \delta, \\
            \frac{\pi\sqrt{\theta}}{2\sqrt{\log{1/\epsilon}}}e^{-\beta/\epsilon}\sin\left(\frac{x + \alpha}{\epsilon}\right),  & \quad
            x \in (-\alpha + \delta, \alpha - \delta), \\
            \frac{\pi\sqrt{\theta}}{2\sqrt{\log{1/\epsilon}}}e^{-\beta/\epsilon}\left[\sin\left(\frac{x + \alpha}{\epsilon}\right)+\sin\left(\frac{x - \alpha}{\epsilon}\right) \right], & \quad x>  \alpha + \delta,
        \end{array}
    \right.\label{1.RNfinal}
\end{equation}
where $\delta = \mathcal{O}(\eps^{1/2})$ describes a neighbourhood of a special line in the complex plane
known as a Stokes curve. The two exponentially small sinusoidal contributions switch on rapidly in this neighbourhood
as the associated Stokes curves are crossed at $x = \pm\alpha$. We will identify the asymptotic approximation \eqref{eps-ode}
by requiring that the solution tends to zero as $x \rightarrow +\infty$.

The particular solution satisfying (\ref{1.RNfinal}) is obtained by requiring that the solution tend to zero as
$x \rightarrow -\infty$, indicating that all exponentially small contributions are zero in this limit.
It is possible to obtain different solutions by imposing conditions such as symmetry about $x = 0$,
or requiring that the solution tend to zero as $x \rightarrow +\infty$.
Each of these choices produces the same result \eqref{eps-ode}.

We begin by expressing $u(x)$ in terms of an asymptotic power series
\begin{equation}
u(x) \sim \sum_{j=0}^{\infty} \epsilon^{2j}u_j(x; \log1/\epsilon),\label{1.series}
\end{equation}
where $u_j$ only contain logarithmic terms in $\epsilon$. In general, including
the logarithmic behaviour requires a nested power series with multiple length scales;
however, this complication can be avoided in the present study
by permitting the series terms $u_j$ to vary logarithmically in $\epsilon$.

By applying the asymptotic series \eqref{1.series} to the governing equation \eqref{ode}, we obtain at leading order
\begin{equation}
\frac{d^2 u_0}{dx^2} + u_0 - \frac{\theta u_0}{1+u_0^2} = 0,
\end{equation}
which is the second-order eqution \eqref{e=0}. We therefore set $u_0(x) = U(x)$,
where $U$ is defined and studied in Section \ref{sec-2}.

By substituting \eqref{1.series} into \eqref{ode} and matching at $\mathcal{O}(\eps^{2j})$
as $\epsilon\rightarrow 0$, we can determine a recurrence relation for the series terms, given by
\begin{equation}
\diff{^4 u_{j-1}}{x^4} + \diff{^2 u_{j}}{x^2} + u_j - \frac{\theta u_j (U^2-1)}{(U^2+1)^2} + \ldots= 0,\qquad j > 1.\label{1.recur}
\end{equation}
The omitted terms are proportional to $u_{j-k}$ with $k > 1$. These terms are smaller
compared to the retained terms in the limit that $j \rightarrow \infty$ due to the divergence
of the asymptotic series \eqref{1.series}. Consequently, these terms do not play a role in
the exponential asymptotic analysis. By Theorem \ref{theorem-001}, $U(x)$ has singularities
in its analytic continuation at $x = \pm \alpha \pm i \beta$, with the signs chosen independently.
We see that obtaining $u_j$ from $u_{j-1}$ requires taking four derivatives of the $u_{j-1}$ term,
and two integrations, This indicates that any singularity in $u_{j-1}$ must also appear in $u_j$,
with a strength that has increased by two. This repeated differentiation causes the series \eqref{1.series} to diverge.

For singularly-perturbed problems, it was observed by \cite{Dingle1} that asymptotic behaviour of the terms
of a divergent asymptotic series obtained by repeated differentiation are given as a sum of factorial-over-power
contributions, containing the most singular terms present at each order of the asymptotic expansion.
Motivated by this observation, we attempt to write the global form of the series terms $u_j$
as a sum of terms with the factorial-over-power expression given by
\begin{equation}
u_j(x) \sim \frac{F(x; \log 1/\epsilon) \Gamma(2 j -1)}{\chi(x)^{2 j-1}}\qquad \mathrm{as}\qquad j \rightarrow \infty,
\label{1.ansatz0}
\end{equation}
where $G$ and $\chi$ are to be defined subject to the condition $\chi(x_0) = 0$, where $x_0$
is a singularity of $U(x)$ nearest to the real axis. Since $U(x)$ has four singularities,
located at $x = \pm \alpha \pm i \beta$, the asymptotic behaviour of the late-order terms $u_j$ is therefore
given by a sum of four factorial-over-power ansatz terms \eqref{1.ansatz0}.

By substituting the ansatz \eqref{1.ansatz0} into the recurrence relation \eqref{1.recur},
it is possible to determine the form of $G$, $\gamma$ and $\chi$ associated with each singularity.
We see that as $j \rightarrow \infty$, $u_j$ is dominant compared to $u_{j-k}$ for $k > 1$.
This confirms that the omitted terms in \eqref{1.recur} will not contribute at any of the orders
required to determine the late-order behaviour of the system. We will perform this analysis
to determine the late-order terms associated with the singularity at $x_0 = -\alpha + i \beta$, and state
the remaining contributions without derivation. In particular, we find at $\mathcal{O}(u_{j+1})$ that
\begin{equation}
\left(\diff{\chi}{x}\right)^2+ 1 = 0, \qquad \chi(x_0) = 0,\label{1.sing1}
\end{equation}
which yields
\begin{equation}
\chi(x) = \pm i (x - x_0).\label{1.chipm}
\end{equation}
We recall that the Stokes phenomenon describes the switching of exponentially small solution components,
and can only occur if $\mathrm{Re}(\chi) > 0$. Therefore, we disregard the negative choice of sign in \eqref{1.chipm}
and write $\chi(x) = i(x-x_0)$.

At $\mathcal{O}(u_{j+1/2})$ we obtain
\begin{equation}
\diff{F}{x} = 0,
\end{equation}
which yields constant $F$. In order to obtain the constant values of $F$ and $\gamma$, we must match
the global behaviour of the late-order ansatz \eqref{1.ansatz0} with the local behaviour of the solution
for $U(x)$ in the neighbourhood of the singularity at $x_0$.

We therefore define a scaled variable $\eta$, defined by $\epsilon \eta = x - x_0$, and
match a local solution in the neighbourhood of the singularity with the inner limit of
the outer solution for the series term ansatz. The technical details of this process
are illustrated in detail in \cite{Daalhuis1}. The asymptotic matching reveals that
\begin{equation}
F = \frac{\sqrt{\theta}}{2\sqrt{\log1/\epsilon}},
\end{equation}
which we present here despite the actual value of $F$ is not used in the subsequent analysis.

Repeating this procedure for the three remaining singularities and adding the results gives as
$j \to \infty$:
\begin{align}
\nonumber u_j(x) \sim \frac{F \Gamma(2j -1)}{[i(x+\alpha-i \beta )]^{2j-1}} +&\frac{F \Gamma(2j -1)}{[-i(x+\alpha+i\beta)]^{2j-1}} \\&+\frac{ F\Gamma(2j -1)}{[i(x-\alpha-i \beta)]^{2j-1}}
+\frac{ F\Gamma(2j -1)}{[-i(x-\alpha + i \beta)]^{2j-1}}.
\label{1.LOT}
\end{align}

Once the late-order terms have been obtained, there exist several methods that may be used
to find the Stokes structure of the solution, and to determine the exponentially small behaviour
that is switched as the Stokes lines are crossed. One can use Borel summation \cite{Bennett1, Berry1, Berry2, Howls1, Howls2}
or matched asymptotic expansions \cite{Chapman1, Daalhuis1} in order to determine the Stokes line contributions.

In both cases, the critical idea is that the divergent asymptotic series may be truncated in an optimal fashion,
which minimizes the approximation error. This optimal truncation point is controlled by the form of
the late-order terms, and may be determined simply from this asymptotic series term behaviour.
We will again concentrate on the contribution due to the singularity at $x_0 = -\alpha + i \beta$.
The corresponding analysis for the remaining contributions is omitted, as they may be obtained in similar fashion.

We truncate the asymptotic series \eqref{1.series} after $N$ terms to obtain
\begin{equation}
u(x) = \sum_{j = 0}^{N-1} \epsilon^{2j} u_j(x) + R_N(x),\label{1.truncated}
\end{equation}
where $R_N$ is the exact remainder after truncation. As is discussed in \cite{Boyd2}, the optimal truncation point
typically occurs at the value of $j$ for which the $j$th term of the asymptotic series is smallest.
We therefore require the value
of $N$ which minimizes $\epsilon^{2N}u_N$. If we assume this occurs after a large number of terms,
we may apply the ansatz \eqref{1.ansatz0} to $u_N$, and then minimize the resultant expression
in order to show that the minimum value is obtained for $N \sim |x+\alpha-i \beta|/2\epsilon$.
We write $N = |x+\alpha-i \beta|/2\epsilon + \omega$, where $0 \leq \omega < 1$, in order
to ensure that $N$ takes integer value.

Substituting the truncated series \eqref{1.truncated} into the governing equation \eqref{ode}
and using the recurrence relation \eqref{1.recur} when necessary, gives as $\epsilon \rightarrow 0$
\begin{equation}
\epsilon^2 \diff{^4 R_N}{x^4} + \diff{^2 R_N}{x^2} + \ldots \sim \epsilon^{2N} \diff{^2 u_N}{x^2},\label{1.remeq}
\end{equation}
where the omitted terms are small in the asymptotic limit.

The solution behaviour for \eqref{1.remeq} in regions where the right-hand side is small,
and the problem may therefore be considered homogeneous, can be obtained using
the Liouville-Green (JWKB) method in the limit that $\epsilon \rightarrow 0$, giving
$R_{N}(x) \sim C  e^{-i (x + \alpha - i\beta)/{\epsilon}}$, where $C$ is some constant.
Importantly, near the Stokes line, the right-hand side of \eqref{1.remeq} will not be negligible,
and this solution is not valid. Consequently, to determine $R_{N}$ near Stokes lines, we write
\begin{equation}
R_{N}(x) \sim \mathcal{S} e^{-i (x + \alpha - i\beta)/{\epsilon}}\qquad \mathrm{as} \qquad \epsilon \rightarrow 0,\label{1.RNswitch}
\end{equation}
where $\mathcal{S}$ is a Stokes multiplier, or a quantity that is constant away from the Stokes line,
but permitted to vary rapidly in the neighbourhood of the Stokes line. The remainder equation \eqref{1.remeq}
becomes, after some simplification
\begin{equation}
\diff{\mathcal{S}}{x} \sim -
\frac{\epsilon^{2N+1} F \Gamma(2N+1)}{4[i(x+\alpha-i\beta)]^{2N+1}}e^{i(x+\alpha-i \beta)/\epsilon},\label{1.Sx}
\end{equation}
Recalling that $N = |x +\alpha-i\beta|/2\epsilon + \omega$, we apply a change of variables,
expressing the singulant in polar coordinates to give $i (x +\alpha - i \beta) = r e^{i \vartheta}$.
Stokes lines typically follow radial directions in this coordinate system, so we restrict our attention
to angular variation. Noting that $N = r/2\epsilon + \omega$, and applying Stirling's formula, we reduce \eqref{1.Sx} to
\begin{equation}
\diff{\mathcal{S}}{\vartheta} \sim \frac{i F \sqrt{2\pi r }}{4\epsilon^{1/2}}\exp\left(\frac{r}{\epsilon}(e^{i \vartheta} - 1) -   i\vartheta\left(\frac{r}{\epsilon}+2 \omega \right)\right)
\end{equation}
as $\epsilon \rightarrow 0$. We see that the right-hand side of this expression is exponentially small,
except on $\vartheta = 0$, across which the Stokes multiplier varies rapidly. This is therefore the Stokes line
associated with the late-order behaviour, and corresponds to $\mathrm{Im}(\chi) = 0$ and $\mathrm{Re}(\chi) > 0$,
as expected. This condition defines a line extending vertically downwards from the singularity at $x_0 = - \alpha + i \beta$
along $\mathrm{Re}(x) = -\alpha$. In order to determine the quantity switched as this Stokes line is crossed,
we apply an inner expansion in the neighbourhood of this curve, given by $\vartheta = \epsilon^{1/2}\phi$. This gives
\begin{equation}
\diff{\mathcal{S}}{\phi} \sim \frac{i  F \sqrt{2\pi r}}{4}e^{-r\phi^2/2}.\label{1.Sdiff}
\end{equation}
We apply the condition that the Stokes contribution is zero as $\mathrm{Re}(x) \rightarrow -\infty$,
which implies that $\mathcal{S}$ is zero on the left-hand side of the Stokes line ($\phi \rightarrow -\infty$).
Solving \eqref{1.Sdiff} with this condition gives
\begin{equation}
\mathcal{S}(\phi) \sim \frac{i F \sqrt{2\pi}}{4}\int_{-\infty}^{\phi/\sqrt{r}} e^{-t^2/2} dt.
\end{equation}
Crossing the Stokes line in the positive $\vartheta$ direction is equivalent to taking the limit as $\phi \rightarrow \infty$.
Hence, as the Stokes line is crossed, the value of $\mathcal{S}$ varies smoothly in a region of
width $\mathcal{O}(\eps^{1/2})$ from zero to $i \mathcal{S}_{\mathrm{on}}$, where $\mathcal{S}_{\mathrm{on}} = { \pi  F}/{2}$.

Hence, using \eqref{1.RNswitch}, the contribution that is switched on across the Stokes line associated
with the singularity at $x_0 = -\alpha + i \beta$, which follows the curve $\mathrm{Re}(x) = -\alpha$, is given by
\begin{equation}
R_N(x) \sim i \mathcal{S}_{\mathrm{on}} e^{-i(x + \alpha-i \beta)}\qquad \mathrm{as} \qquad \epsilon \rightarrow 0.
\end{equation}
Using similar analysis, we find that the Stokes switching contribution associated with the singularity at
$\bar{x}_0 = -\alpha - i \beta$, which also follows the curve $\mathrm{Re}(x) = -\alpha$,
is given by the complex conjugate of this expression. Furthermore the contribution that is switched
on across the Stokes line associated with the singularity at $x_0^* = \alpha + i \beta$,
which follows the curve $\mathrm{Re}(x) = \alpha$ is given by
\begin{equation}
R_N(x) \sim - i \mathcal{S}_{\mathrm{on}}e^{i(x - \alpha-i \beta)}\qquad \mathrm{as} \qquad \epsilon \rightarrow 0,
\end{equation}
while the contribution associated with the singularity at $\overline{x}^*_0 = \alpha - i \beta$,
which is also switched on across the Stokes line $\mathrm{Re}(x) = \alpha$, takes the corresponding
conjugate behaviour. Combining the four contributions gives the composite exponentially small behaviour $R_N$ as
\begin{equation}
R_N(x) \sim i \mathcal{S}_1 e^{-\beta/\epsilon} (e^{-i(x+\alpha)} - e^{i(x+\alpha)}) +
i \mathcal{S}_2 e^{-\beta/\epsilon} (e^{-i(x-\alpha)} - e^{i(x-\alpha)}),
\end{equation}
where $\mathcal{S}_1$ switches rapidly from zero to $\mathcal{S}_{\mathrm{on}}$ in a region of width $\mathcal{O}(\eps^{1/2})$ about the Stokes line $\mathrm{Re}(x) = -\alpha$, while $\mathcal{S}_2$ switches from zero to $\mathcal{S}_{\mathrm{on}}$ about the Stokes line $\mathrm{Re}(x) = \alpha$.

\begin{figure}
\centering
\begin{tikzpicture}
[xscale=0.75,>=stealth,yscale=0.75]

\fill[black,opacity=0.075] (2,2.75) -- (-2,2.75) -- (-2,-2.75) -- (2,-2.75) -- cycle;
\fill[black,opacity=0.2] (2,2.75) -- (2,-2.75) -- (4.5,-2.75) -- (4.5,2.75) -- cycle;

\draw[decoration = {zigzag,segment length = 1.5mm, amplitude = 0.35mm},decorate, line width=0.35mm,black] (2,1.5)--(2,2.75);
\draw[decoration = {zigzag,segment length = 1.5mm, amplitude = 0.35mm},decorate, line width=0.35mm,black] (-2,1.5)--(-2,2.75);
\draw[decoration = {zigzag,segment length = 1.5mm, amplitude = 0.35mm},decorate, line width=0.35mm,black] (-2,-1.5)--(-2,-2.75);
\draw[decoration = {zigzag,segment length = 1.5mm, amplitude = 0.35mm},decorate, line width=0.35mm,black] (2,-1.5)--(2,-2.75);

\draw[gray,line width=0.65mm] (2,1.5) -- (2,-1.5);
\draw[gray,line width=0.65mm] (-2,1.5) -- (-2,-1.5);

\draw[line width=0.75mm] (1.85,1.35) -- (2.15,1.65);
\draw[line width=0.75mm] (2.15,1.35) -- (1.85,1.65);
\draw[line width=0.75mm] (1.85,-1.35) -- (2.15,-1.65);
\draw[line width=0.75mm] (2.15,-1.35) -- (1.85,-1.65);
\draw[line width=0.75mm] (-1.85,1.35) -- (-2.15,1.65);
\draw[line width=0.75mm] (-2.15,1.35) -- (-1.85,1.65);
\draw[line width=0.75mm] (-1.85,-1.35) -- (-2.15,-1.65);
\draw[line width=0.75mm] (-2.15,-1.35) -- (-1.85,-1.65);

\node at (2.2,1.5) [above right] {\scriptsize{$\alpha + i \beta$}};
\node at (-2.2,1.5) [above left] {\scriptsize{$-\alpha + i \beta$}};
\node at (-2.2,-1.5) [below left] {\scriptsize{$-\alpha - i \beta$}};
\node at (2.2,-1.5) [below right] {\scriptsize{$\alpha - i \beta$}};

\draw[->] (-5,0) -- (5,0) node[above] {\scriptsize{$\mathrm{Re}(x)$}};
\draw[->] (0,-3) -- (0,3) node[above] {\scriptsize{$\mathrm{Im}(x)$}};


\draw[thick] (6.25,0.25-0.25) -- (6.75,0.25-0.25) -- (6.75,0.75-0.25) -- (6.25,0.75-0.25) -- cycle;
\fill[black,opacity=0.075] (6.25,-0.25-0.25) -- (6.75,-0.25-0.25) -- (6.75,-0.75-0.25) -- (6.25,-0.75-0.25) -- cycle;
\draw[thick] (6.25,-0.25-0.25) -- (6.75,-0.25-0.25) -- (6.75,-0.75-0.25) -- (6.25,-0.75-0.25) -- cycle;
\fill[black,opacity=0.2] (6.25,-1.25-0.25) -- (6.75,-1.25-0.25) -- (6.75,-1.75-0.25) -- (6.25,-1.75-0.25) -- cycle;
\draw[thick] (6.25,-1.25-0.25) -- (6.75,-1.25-0.25) -- (6.75,-1.75-0.25) -- (6.25,-1.75-0.25) -- cycle;

\draw[gray,line width=0.65mm] (6.25,1.5-0.25) -- (6.75,1.5-0.25);
\draw[decoration = {zigzag,segment length = 1.5mm, amplitude = 0.35mm},decorate, line width=0.35mm,black]  (6.25,2.25-0.25) -- (6.75,2.25-0.25);
\draw[line width=0.75mm] (6.35,3.15-0.25) -- (6.65,2.85-0.25);
\draw[line width=0.75mm] (6.35,2.85-0.25) -- (6.65,3.15-0.25);

\node at (7,1.5-0.25) [right] {Stokes line};
\node at (7,2.25-0.25) [right] {Branch cut};
\node at (7,3-0.25) [right] {Singularity};

\node at (7,0.5-0.25) [right] {$R_N = 0$};

\node at (7,-0.5-0.25) [right] {$R_N \sim \pi  F e^{-\beta/\epsilon}\sin\left(\tfrac{x-\alpha}{\epsilon}\right) $};
\node at (7,-1.5-0.25) [right] {$R_N \sim \pi  F  e^{-\beta/\epsilon}\sin\left(\tfrac{x-\alpha}{\epsilon}\right) $};
\node at (9,-2.5-0.25) [right] {$+\pi \epsilon F  e^{-\beta/\epsilon}\sin\left(\tfrac{x+\alpha}{\epsilon}\right) $};

\end{tikzpicture}
\caption{Complete Stokes structure for $u(x)$. In the region $\mathrm{Re}(x) < -\alpha$,
there are no exponentially small oscillations. In the region $-\alpha < \mathrm{Re}(x) < \alpha$,
there is one oscillatory wave contribution. In the region $\mathrm{Re}(x) > \alpha$, there are
two oscillatory contributions. Each oscillatory contribution is switched on smoothly but rapidly
across the Stokes lines, which are depicted as thick gray lines at $\mathrm{Re}(x) = -\alpha$ and $\mathrm{Re}(x) = \alpha$.}
\label{F:StokesFig1}
\end{figure}
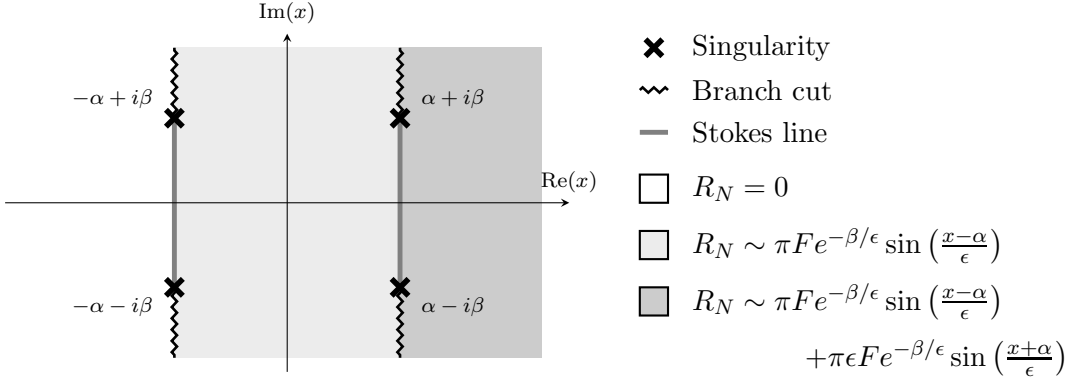

It is simple to rewrite this in terms of real-valued trigonometric functions, giving
\begin{equation}
R_N(x) \sim 2\mathcal{S}_1 e^{-\beta/\epsilon} \sin\left(\frac{x-\alpha}{\epsilon}\right) + 2\mathcal{S}_2 e^{-\beta/\epsilon} \sin\left(\frac{x+\alpha}{\epsilon}\right) ,\label{2.RNcos}
\end{equation}
This gives the asymptotic expression given in \eqref{1.RNfinal} and illustrated in Figure \ref{F:StokesFig1}.
The asymptotic result (\ref{eps-ode}) is recovered by considering the behaviour of solutions
in the region of the complex plane in which both oscillatory contributions have been switched on,
or $\mathrm{Re}(x) > \alpha$. It is clear from \eqref{1.LOT} that in the region
$\mathrm{Re}(x) > \alpha$, the oscillatory contribution may be rewritten as
\begin{equation}
R_{N}(x) \sim 2\pi  F e^{-\beta/\epsilon}\cos\left(\frac{\alpha}{\epsilon}\right) \sin\left(\frac{x}{\epsilon}\right) ,\qquad \mathrm{as} \qquad \epsilon \rightarrow 0.\label{1.RTOT2}
\end{equation}
When written in this form, it is clear that $R_{N}$ cancels in this region $\mathrm{Re}(x) > \alpha$ if
$\alpha/\epsilon = {\pi (2m-1)}/{2}$ with $m \in \mathbb{N}$. If we denote these choices of the small parameter
as $\epsilon_m$, this yields the asymptotic result \eqref{eps-ode}, corresponding to transparent points. We note that these transparent points are approximate solitary wave solutions, as we only demonstrated cancellation of the dominant contributions arising from \eqref{2.dchi} associated with $M = \pm1$. This is unlike the fourth-order equation, for which we found parameter values that cause that all oscillatory contributions to the solution to vanish, thereby producing solitary wave solutions.

\subsection{Numerical results}
\label{NumResults}

We confirm numerically the validity of the asymptotic formula (\ref{eps-ode})
by computing solutions to the fourth-order equation ~({\ref{ode}) that decays to zero at infinity.

Define a dynamical system in the phase space $(u,u',u'',u''')$ associated with the fourth-order
equation ~({\ref{ode}). The only equilibrium is $O=(0,0,0,0)$, hence the solitary wave solutions correspond to
homoclinic loops to this equilibrium. The system is conservative due to the first integral
\begin{gather}
\mathcal{E} = \varepsilon^2 \left[ 2 \frac{d^3 u}{dx^3} \frac{du}{dx} - \left( \frac{d^2 u}{dx^2} \right) ^2\right] +
\left( \frac{d u}{dx} \right)^2 + u^2 -\theta\log(1+u^2),\label{FirstInt}
\end{gather}
which generalizes the first integral (\ref{first-order}) of the second-order equation (\ref{e=0}).
The fourth-order equation (\ref{ode}) is invariant with respect to the transformation $x\to-x$,
therefore the dynamical system is invariant with respect to the involution
\begin{gather*}
\sigma_1:\quad (u,u',u'',u''')\quad \to\quad (u,-u',u'',-u''').
\end{gather*}
The invariant set of $\sigma_1$ is the 2D plane $S=\{u'=0, u'''=0\}$.
Since Eq.~(\ref{ode}) is also invariant with respect to the transformation $u\to -u$,
the dynamical system is invariant with respect to another involution
\begin{gather*}
\sigma_2:\quad (u,u',u'',u''')\quad \to\quad (-u,-u',-u'',-u''').
\end{gather*}

The equilibrium $O$ lies in $\mathcal{E}_0=\{\mathcal{E}=0\}$, the zero level of the first integral.
Evidently, $\mathcal{E}_0$ is a 3D set.  The four eigenvalues in the linearization of
the dynamical system at $O$ are given by two pairs $\pm \lambda_1$ and $\pm \lambda_2$,
where
\begin{align*}
\lambda_{1}=\frac1{\sqrt{2}\varepsilon} \sqrt{ \sqrt{1+4\varepsilon^2(\theta-1)} - 1},\quad
\lambda_{2}=\frac i{\sqrt{2}\varepsilon} \sqrt{ \sqrt{1+4\varepsilon^2(\theta-1)} + 1}.
\end{align*}
Therefore, $O$ is classified as the saddle-center point for any $\varepsilon$ and $\theta>1$.
This implies that there exist a pair of outgoing trajectories $\gamma^+_{1,2}$ of $O$ and a pair of incoming trajectories $\gamma^-_{1,2}$ of $O$. The pair of trajectories $\gamma^+_{1}$ and $\gamma^+_{2}$, as well as $\gamma^-_{1}$ and $\gamma^-_{2}$, are related with each other by the involution $\sigma_2$. Similarly, the trajectories $\gamma^+_{1}$ and $\gamma^-_{1}$, as well as $\gamma^+_{2}$ and $\gamma^-_{2}$ are related by the involution $\sigma_1$.

All the trajectories $\gamma^{\pm}_{1,2}$ with the connection to $O$ are also situated at the zero energy level $\mathcal{E}_0$.
The homoclinic orbit arises due to an intersection of $\gamma^+_{1}$ (or $\gamma^+_{2}$) with $\gamma^-_{1}$ (or $\gamma^-_{2}$).
The intersection of two trajectories within the 3D set $\mathcal{E}_0$ does not correspond to the generic case,
hence the homoclinic orbits are not generic. However, homoclinic orbits may exist for selected values
of the governing parameter $\varepsilon$ as a result of co-dimension one bifurcations.

In what follows we restrict the consideration by {\it even solutions} to
the fourth-order equation (\ref{ode}). They correspond to the homoclinic orbits of $O$ that are invariant with respect to the involution $\sigma_1$. In order to compute these orbits and the corresponding values of $\varepsilon$ we make use of the fact that a symmetric homoclinic orbit in a reversible system must intersect the invariant set of the involution (see, e.g., Lemma 3 in \cite{VF92}). Hence
the trajectory $\gamma^+_{1}$ (or $\gamma^+_{2}$) has to cross the plane $S$. Then $\gamma^-_{1}$ (or $\gamma^-_{2}$)
also crosses $S$ at the same point and the homoclinic loop is composed from the two pieces of these trajectories
before they hit the plane $S$. Since $\gamma^+_1$ and $\gamma^+_2$ are related by the involution $\sigma_2$, it is
sufficient to consider the trajectory $\gamma^+_1$ only and to detect numerically its intersections with the plane $S$.

\begin{figure}
{\centerline{\includegraphics [scale=0.6]{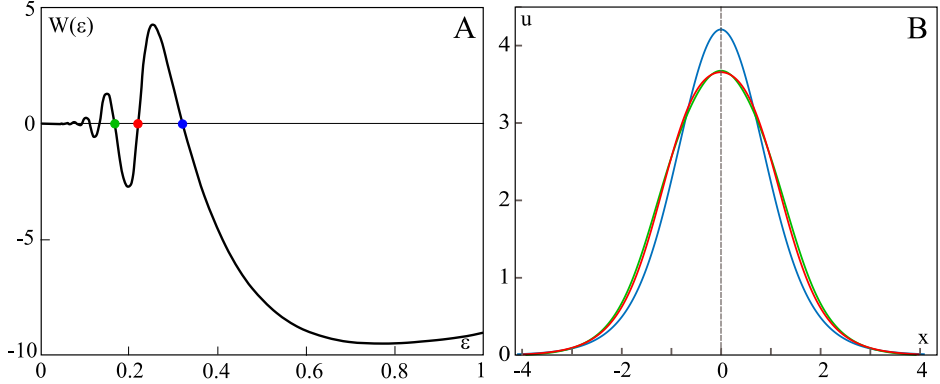}}} \caption{A: the plot of $W(\varepsilon)$ for $\theta=5$. By blue, red and green balls three greatest zeros of $W(\varepsilon)$ are shown: $\varepsilon_1\approx 0.32128$, $\varepsilon_2\approx 0.22152$ and $\varepsilon_3\approx 0.16684$. B: the solution profiles corresponding to $\varepsilon_{1}$, $\varepsilon_{2}$, $\varepsilon_{3}$. The colors (blue, red and green) of the profiles corresponds to the colors of the balls in panel A.} \label{W_Theta_5}
\end{figure}

\begin{figure}
{\centerline{\includegraphics [scale=0.6]{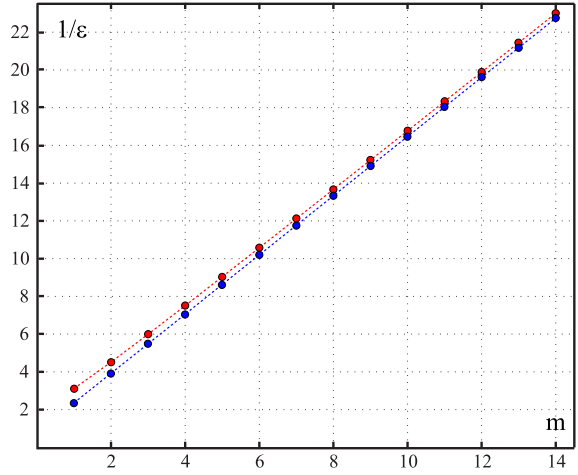}}} \caption{The numerical values of $\varepsilon_m^{-1}$
corresponding to homoclinic orbits (red balls) and their asymptotic values (blue balls) computed from
(\ref{spacing-ode}) for $\theta = 5$. } \label{C_Theta_5}
\end{figure}

For a given value of $\varepsilon$ we compute the trajectory $\gamma^+_1$ until the first point $P_0$ where $u' |_{P_0} =0$ and register the values
$W := u'''|_{P_0}$. Then we vary the value $\varepsilon$ and plot $W$ versus $\varepsilon$.
This plot for $\theta=5$ is shown in Fig.~\ref{W_Theta_5}, panel A. We can see that $W$ oscillates in $\varepsilon$
and has many zeros. For each zero of $W$, both $u' |_{P_0}$ and $u''' |_{P_0}$ vanishes, hence
$\gamma^+_1$ intersects $S$ at $P_0$ and represent the homoclinic orbit.

The numerical computation of $\gamma^+_1$ starts from  a vicinity of $O$ where
the components $u,u',u'',u'''$ are small. Then $\gamma^+_1$ can be extended
to larger values of $u,u',u'',u'''$ by means of the fourth-order Runge--Kutta method.
Fig.\ref{W_Theta_5}, panel B, represents three profiles of the solitons corresponding to three largest
zeros of $W$ at $\varepsilon_1\approx 0.32128$, $\varepsilon_2\approx 0.22152$ and $\varepsilon_3\approx 0.16684$.
More values of $\varepsilon$ for which $W$ is zero are shown in Table \ref{table:5}. It follows from Table \ref{table:5} that
the values $\{\varepsilon_m^{-1} \}_{m \in \mathbb{N}}$ are asymptotically equidistant with spacing
close to $\pi/\alpha$ where $\alpha$ is the real part of the singularity of $U(z)$.
For $\theta=5$, we have detected numerically $\alpha\approx2.003$, therefore $\pi/ \alpha \approx 1.56834$,
which is close to the numerical values in Table \ref{table:5}.

Fig.~\ref{C_Theta_5} presents the values $\{\varepsilon_m\}_{m \in \mathbb{N}}$
computed numerically and from the asymptotic formula (\ref{spacing-ode}).
The correspondence is fairly good. Similar agreement is observed for other values of $\theta$.

\begin{table}
\centerline{
\begin{tabular}{c|c|c|c}
 $m$ & $\varepsilon_m = \frac{2 \alpha}{(2m-1) \pi}$ & Computed $\varepsilon_m$ & $\varepsilon_m^{-1}-\varepsilon_{m-1}^{-1}$ \\ \hline
 1 & 0.42505 & 0.32128 & \\
				  2 & 0.25503 & 0.22152 & 1.40163\\
				  3 & 0.18216 & 0.16684 & 1.47497\\
				  4 & 0.14168 & 0.13322 & 1.51259\\
							  \vdots & \vdots & \vdots & \vdots\\
							  12& 0.05101 & 0.05029 & 1.55773\\
							  13& 0.04723 & 0.04663 & 1.55911\\
							  14& 0.04397 & 0.04347 & 1.56117\\
\end{tabular}
}
\caption{The values $\varepsilon$ corresponding to the homoclinic orbits of Eq. (\ref{ode-intro}) at $\theta = 5$.}
\label{table:5}
\end{table}

\section{Approximate solitary wave solutions to the advance-delay equation}

Here we consider the advance-delay equation:
\begin{eqnarray}
\frac{1}{h^2} \left[ u(x+h) - 2 u(x) + u(x-h) \right] + u(x) - \frac{\theta u(x)}{1+u(x)^2}=0,
\label{advance-delay}
\end{eqnarray}
where $h$ is a small positive parameter for the lattice spacing.
The main goal of this section is to show the existence of a countable sequence of approximate
solitary wave solutions to the advance-delay equation (\ref{advance-delay}) at $h$ near $\{ h_m \}_{m \in \mathbb{N}}$,
where the sequence $\{ h_m \}_{m \in \mathbb{N}}$ accumulates to zero as $m \to \infty$
according to the asymptotic representation:
\begin{equation}
\label{eps-advance-delay}
h_m \sim \frac{4 \alpha}{(2m-1)}, \quad m \in \mathbb{N},
\end{equation}
where $\alpha > 0$ is defined by (\ref{alpha-beta}).
If we use $h_m = 2 \sqrt{3} \varepsilon_m$
according to the correspondence (\ref{h-eps}),
then the spacing between two consequent values of $\{ \varepsilon_m\}_{m \in \mathbb{N}}$
is asymptotically given by
\begin{equation}
\label{spacing-advance-delay}
\frac{1}{\varepsilon_{m+1}} - \frac{1}{\varepsilon_m} \to \frac{\sqrt{3}}{\alpha} \quad \mbox{\rm as} \quad m \to \infty,
\end{equation}
which is different from the asymptotic result (\ref{spacing-ode}) for the fourth-order equation (\ref{ode}).

Approximate solitary wave solutions are again obtained by two equivalent methods
in Sections \ref{sec-4-1} and \ref{sec-4-2}. These approximate solutions are related to the
transparent points in Definition \ref{def-1}, which are computed numerically in Section \ref{sec-4-3}.

\subsection{Analysis of oscillatory integrals}
\label{sec-4-1}

Let $U$ be the even,  positive, and exponentially decaying solution
to the second-order equation (\ref{e=0}) defined in the implicit form by (\ref{exact}).
We are looking for a symmetric solution to the advance-delay equation (\ref{advance-delay}) in the perturbed
form $u = U + v$. Substitution yields the persistence problem for $v$:
\begin{equation}
\label{perturbed-problem-advance-delay}
L_h v = H_h + N(v),
\end{equation}
where $N(v)$ is the same as in (\ref{nonlinear-term-N}),
$L_h$ is a new linearization operator at the zero solution given by
\begin{equation}
\label{operator-L-advance-delay}
(L_h v)(x):= -\frac{1}{h^2} \left[ v(x+h) - 2 v(x) + v(x-h) \right] + (\theta - 1) v(x),
\end{equation}
and $H_h$ is a new source term given by
\begin{equation}
\label{source-term-H-advance-delay}
H_h := \frac{1}{h^2} \left[ U(x+h) - 2 U(x) + U(x-h) \right] - \frac{d^2 U}{d x^2}.
\end{equation}
Fourier transform for the operator $L_h$ yields the dispersion relation:
\begin{equation}
\label{DR}
D_h(k) := \frac{4}{h^2} \sin^2 \left( \frac{kh}{2}\right) + \theta - 1, \quad k \in \mathbb{R}.
\end{equation}
If $\theta > 1$, there exist no real roots of the transcendental equation $D_h(k) = 0$.
However, as $h \to 0$, there exists a countable sequence of roots
at $k_n = 2\pi n h^{-1} + \mathcal{O}(1)$, $n \in \mathbb{N}$,
where the $\mathcal{O}(1)$ correction is purely imaginary.

Although the dispersion relation $D_h(k) = 0$ does not exhibit real roots in $k$,
the inverse of $L_h$ on $L^2(\mathbb{R})$ is bounded but singular as $h \to 0$.
As a result, iterations for the fixed-point problem (\ref{perturbed-problem-advance-delay})
do not converge to a unique fixed point unless a countable number of solvability conditions is added.
For the solution of the linear inhomogeneous equation $L_h v = H_h$, the set of
solvability conditions is given by $\{I_h^{(n)} = 0\}_{n \in \mathbb{N}}$, where
\begin{equation}
\label{integral-I-advance-delay}
I_h^{(n)} := \int_{\mathbb{R}} H_h(x) e^{i k_n x} dx, \quad n \in \mathbb{N}.
\end{equation}
By change of variables and integration by parts, these integrals become
\begin{equation}
\label{I-n}
I_h^{(n)} = \left[ k_n^2 - \frac{4}{h^2} \sin^2\left( \frac{k_n h}{2}\right) \right] I(k_n),
\end{equation}
where $I(\varkappa)$ with $\varkappa = k_n$ is given by (\ref{FT-of-U}). Since
$k_n = 2 \pi n h^{-1} + \mathcal{O}(1)$ as $h \to 0$, substituting
the asymptotic approximation (\ref{I_varkappa}) into (\ref{I-n}) yields
the asymptotic result:
\begin{eqnarray}
I_h^{(n)} \sim \frac{2\pi\sqrt{\theta}}{\sqrt{\log (2 \pi n h^{-1})}}e^{-2\pi n \beta h^{-1}}\cos(2\pi n \alpha h^{-1}).
\label{main-advance-delay}
\end{eqnarray}
Since $\{ I_h^{(n)} \}_{n \in \mathbb{N}}$ forms a hierarchic sequence of
exponentially small terms, the dominant contribution is given by $I_h^{(1)}$.
The leading order of $I_h^{(1)}$ vanishes at $\{ h_m \}_{m \in \mathbb{N}}$
given by (\ref{eps-advance-delay}). This asymptotic computation defines
an approximate solitary wave solution to the advance-delay equation (\ref{advance-delay})
for $h$ near $\{ h_m \}_{m \in \mathbb{N}}$
\footnote{There are infinitely many zeros of the dispersion relation $D_h(k) = 0$ in $k$
and $k_1$ is only the smallest root. Even if $I_h^{(1)}$ vanishes at $h_m$, we know
from (\ref{main-advance-delay}) that $I_h^{(2)}$ does not vanish at this $h_m$, therefore,
we cannot predict that the continuous solutions $u \in C(\mathbb{R})$
to the advance-delay equation (\ref{advance-delay}) exist for $h$ near this $h_m$.
The only exception is the value $h_1 = \sqrt{2}$, for which
reduction of the advance-delay equation (\ref{advance-delay}) to the integrable
AL lattice yields an exact solution for $u \in C(\mathbb{R})$.}.

\subsection{Beyond-all-order asymptotics}
\label{sec-4-2}

By studying the advance-delay equation (\ref{advance-delay}) using beyond-all-order methods,
we can recover the asymptotic result \eqref{eps-advance-delay}. We will show that the asymptotic
solution again contains two Stokes lines, each of which switches on an exponentially small contribution
which does not decay in the far field. Approximate solitary wave solutions are associated with the special cases
in which the two contributions cancel.

We will establish that the exponentially small oscillations present in the solution $u(x)$, denoted by $u_{\mathrm{osc}}(x)$,
have the following asymptotic behaviour as $h \rightarrow 0$:
\begin{equation*}
u_{\mathrm{osc}}(x) \sim \left\{
        \begin{array}{ll}
            0, & x < -\alpha - \delta, \\
-\frac{8 \pi^3 \beta^2 F}{h} e^{-2\pi\beta/h}\sin\left(\frac{2\pi(x-\alpha)}{h}\right),  &
x \in (-\alpha + \delta,\alpha - \delta), \\
-\frac{8 \pi^3 \beta^2 F}{h} e^{-2\pi\beta/h}\left[\sin\left(\frac{2\pi(x-\alpha)}{h}\right)+\sin\left(\frac{2\pi(x+\alpha)}{h}\right) \right], &
 x > \alpha + \delta,
        \end{array}
    \right.\label{2.RNfinal}
\end{equation*}
where $\delta$ plays the same role as in \eqref{1.RNfinal}, and $F$ is a constant that is proportional to $1/\sqrt{\log{1/h}}$.
The two exponentially small sinusoidal contributions
switch on rapidly as the associated Stokes curves are crossed at $x = \pm\alpha$ respectively.

This analysis differs in some technical details from the fourth-order equation (\ref{ode})
due to the difference terms. We therefore follow the method established in
for differential-difference equations in \cite{King}, and subsequently utilised
for difference equations in \cite{Joshi1, Joshi2}.

We first apply a Taylor expansion about $h=0$ to smooth solutions of
the advance-delay equation \eqref{advance-delay}, giving
\begin{equation}
\frac{2}{h^2}\sum_{r=1}^{\infty} \frac{h^{2r}}{(2r)!} u^{(2r)} +  u - \frac{\theta u}{1+u^2} = 0,\label{2.ConEq}
\end{equation}	
where $u^{(r)}$ represents the $r$th derivative of $u(x)$ with respect to $x$. This is a differential equation
with infinite order, unlike the fourth-order equation (\ref{ode}). We expand $u$ as a power series, giving
\begin{equation}
u(x) \sim  \sum_{j=0}^{\infty} h^{2j} u_j(x;\log 1/h).\label{2.series}
\end{equation}
Applying this series to \eqref{2.ConEq} and matching at leading order gives the second-order equation \eqref{e=0}.
We therefore set again $u_0(x) = U(x)$, where $U$ is defined and studied in Section \ref{sec-2}.

Matching in the small $h$ limit at $\mathcal{O}(h^{2j})$ we obtain a recurrence relation for $k \geq 0$,
\begin{equation}
2\sum_{r=1}^{j+1}\frac{u_{j+1-r}^{(2r)}}{(2r)!} + u_{j}  - \frac{\theta u_{j} (U^2-1)}{(U^2+1)^2} + \ldots = 0,\label{2.recur}
\end{equation}
where the omitted terms are smaller than those terms retained in the limit that $j \rightarrow \infty$.
As in the analysis of the fourth-order equation (\ref{ode}), we may use this recurrence relation
to determine the asymptotic form of the series terms $u_j$ in the limit that $j \rightarrow \infty$.

In order to determine the late-order terms, we require an ansatz with similar form to \eqref{1.ansatz0}, however
the choice is made more complicated by the observation that the number of contributing terms grows as $r$ increases.
We therefore again apply the recursion relation \eqref{2.recur} to the leading-order solution in the neighbourhood
of the singularity at $x_0$ and obtain:
\begin{equation}
u_j(x) \sim \frac{F(x;\log{1/h})\Gamma(2j)}{\chi(x)^{2j-1}}\qquad \mathrm{as} \qquad j \rightarrow \infty,\label{2.ansatz}
\end{equation}
where $F$ and $\chi$ are to be defined and $\chi(x_0) = 0$. We note that the representation (\ref{2.ansatz})
differs from \eqref{1.ansatz0}, as the argument of the gamma function and the power of the singulant are no longer identical, due to the presence of the summation expression in \eqref{2.recur}, that introduces new terms into the expression for $u_j$ at each recursion.

Putting \eqref{2.ansatz} into \eqref{2.recur} and matching at $\mathcal{O}(u_{k})$, we obtain
\begin{equation}\label{2.finitesum}
2\sum_{r=1}^{k} \frac{(2k-2r-1)}{(2r)!} \left( -\frac{d\chi}{dx} \right)^{2r}= 0.
\end{equation}
Now, as late-order terms are only valid for $k$ being large, it is possible to show that
we introduce only exponentially small error into $\chi$ by taking the behaviour of this equation
as $k \rightarrow \infty$. We evaluate the finite sum, and take the leading-order behaviour in this limit.
Recalling that $\chi(x_0) = 0$ at the singular point $x_0$, this gives
\begin{equation}
\cosh \left( \frac{d\chi}{dx} \right) = 1,\qquad \chi(x_0) = 0.\label{2.dchi}
\end{equation}
This expression is easily solved to give $\chi(x) = 2 \pi i M (x - x_0)$, where $M\in\mathbb{Z}$.
Due to the form of the late-order ansatz \eqref{2.ansatz}, the dominant behaviour must associated
with nonzero values of $\chi$ that have smallest magnitude on the real axis, associated with $M = \pm 1$.
We therefore have
\begin{equation}\label{2.sing1}
\chi(x) = \pm 2 \pi i (x - x_0).
\end{equation}
As in the previous case, for each singularity, we will have one choice of $\chi$ that induces Stokes switching,
which yields $\chi(x) = 2\pi i(x - x_0)$ for $x_0 = -\alpha + i \beta$.

Putting \eqref{2.ansatz} into \eqref{2.recur} and matching at $\mathcal{O}(u_{k-1/2})$ gives
\begin{equation}\label{2.finitesumb}
2\sum_{r=1}^{k} \frac{(2k-2r-1)}{(2r-1)!} \left( -\frac{d\chi}{dx} \right)^{2r-1} \frac{dF}{dx} = 0,
\end{equation}
which is solved to leading-order in the limit that $k \rightarrow 0$, giving
\begin{equation}
\left[ \frac{d\chi}{dx} \cosh\left( \frac{d\chi}{dx} \right) + 2 \sinh\left( \frac{d\chi}{dx} \right) \right] \frac{dF}{dx} =
2\pi i \frac{dF}{dx} = 0.
\end{equation}
Consequently, we know that $F$ is constant. This constant may be determined using asymptotic matching in the same fashion as the fourth-order equation. This is a more complicated process for discrete problems, due to the complexity of the expression (see, for example, \cite{Joshi1, Joshi2}). Performing this analysis reveals that $F$ is a real constant proportional to $1/\sqrt{\log{1/h}}$, which can be obtained numerically. This analysis also validates the choice of ansatz \eqref{2.ansatz}. Furthermore, this constant is identical for each singularity.

Adding all four singularity contributions
and leaving the constant $F$ in the general form gives as $j \rightarrow \infty$,
\begin{align}
\nonumber u_j(x) \sim  \frac{F \Gamma(2j)}{[2 \pi i(x+\alpha-i \beta )]^{2j-1}} + &
\frac{F \Gamma(2j)}{[-2 \pi i(x+\alpha+i\beta)]^{2j-1}} \\&
+ \frac{F \Gamma(2j )}{[2\pi i(x-\alpha-i \beta)]^{2j-1}} +
\frac{F \Gamma(2j)}{[-2\pi i(x-\alpha + i \beta)]^{2j-1}}.\label{2.LOT}
\end{align}

We again determine the exponential contribution associated with the singularity at $x_0 = -\alpha + \beta i$.
We again truncate the asymptotic series \eqref{2.series} after $N$ terms and show that the optimal truncation point
is $N \sim \pi | x+\alpha-i \beta|/h$. We again write $N = \pi|x- \alpha - \beta i|/h + \omega$,
where $0 \leq \omega < 1$, in order to ensure that $N$ takes integer value, and denote the remainder term by $R_N$.
Substituting the truncated series into the governing equation \eqref{ode}, using the recurrence relation \eqref{2.recur}
when necessary, gives as $\epsilon \rightarrow 0$
\begin{equation}
\frac{2}{h^2}\sum_{r=0}^{\infty} \frac{h^{2r}}{(2r)!} R_N^{(2r)} + \ldots \sim h^{2N} \frac{d^2 u_N}{d x^2},\label{2.remeq}
\end{equation}
where the omitted terms are small in the asymptotic limit. Using the Liouville-Green (JWKB) method
on the homogeneous version of \eqref{2.remeq} gives the behaviour away from the Stokes line as
\begin{equation}
R_{N}(x) \sim (Ax + B)  e^{-2 \pi i (x + \alpha - i\beta)/{h}},\qquad \mathrm{as} \qquad \epsilon \rightarrow 0,\label{2.RNhom}
\end{equation}
where $A$ and $B$ are constants. It is clear from the boundary conditions of the problem that $A = 0$; however, we must determine $B$ using asymptotic matching. Had we not determined the value of $A$ here, it would have been obtained as part of the matching condition.
We set
\begin{equation}
R_{N}(x) \sim \mathcal{S} e^{-2 \pi i (x + \alpha - i\beta)/{h}}\qquad \mathrm{as} \qquad \epsilon \rightarrow 0.\label{2.RNswitch}
\end{equation}
where $\mathcal{S}$ is a Stokes multiplier. The remainder equation \eqref{2.remeq} becomes, after some simplification
\begin{equation}
2\sum_{r = 1}^{\infty} \frac{(-2\pi i)^{2r-2}}{(2r-2)!}\diff{^2\mathcal{S}}{x^2} e^{-2\pi i(x+\alpha-i \beta)/h} \sim
h^{2N} \frac{d^2 u_N}{d x^2}.\label{2.Sx0}
\end{equation}
By evaluating the series and applying the late-order ansatz, we obtain
\begin{equation}
\diff{^2\mathcal{S}}{x^2} \sim  h^{2N} \frac{(-2\pi i)^2 (2N-1) F \Gamma(2N+1)}{2 [2\pi i(x+\alpha-i\beta)]^{2N+1}}e^{2\pi i(x+\alpha-i \beta)/h},\label{2.Sx}
\end{equation}
Recalling that $N = |x +\alpha-i\beta|/2h + \omega$, we apply a change of variables, expressing the singulant
in polar coordinates to give $2 \pi i (x +\alpha - i \beta) = r e^{i \vartheta}$. Stokes lines typically follow
radial directions in this coordinate system, so we restrict our attention to variation in angle.
Calculating the variation in the angular direction, noting that $N = r/2h + \omega$, and applying Stirling's formula,
we are able to reduce \eqref{2.Sx} to
\begin{equation}
\diff{\mathcal{S}}{\vartheta} + i\diff{^2\mathcal{S}}{\vartheta^2}  \sim
-\frac{i r^{5/2} \sqrt{2\pi} F}{2 h^{3/2}}\exp\left(\frac{r}{h}(e^{i \vartheta} - 1)
+   i\vartheta\left(\frac{r}{h}+2 \omega +1\right)\right)
\end{equation}
as $\epsilon \rightarrow 0$. We see that the right-hand side of this expression is exponentially small,
except on $\vartheta = 0$, across which the Stokes multiplier varies rapidly. This is therefore the Stokes line
associated with the late-order behaviour, and corresponds to $\mathrm{Im}(\chi) = 0$ and $\mathrm{Re}(\chi) > 0$,
as expected. This condition defines a line extending vertically downwards from the singularity at $x_0 = - \alpha + i \beta$
along $\mathrm{Re}(x) = -\alpha$. In order to determine the quantity switched as this Stokes line is crossed,
we apply an inner expansion in the neighbourhood of this curve, given by $\vartheta = h^{1/2}\phi$. This gives
\begin{equation}
\diff{\mathcal{S}}{\phi} \sim -\frac{i r^{5/2}  \sqrt{2\pi} F}{2h} e^{-r\phi^2/2}.\label{2.Sdiff}
\end{equation}
As before, we find
\begin{equation}
\mathcal{S} \sim -\frac{i r^2 \sqrt{2\pi} F}{2 h} \int_{-\infty}^{\phi/\sqrt{r}} e^{-t^2/2} dt.
\end{equation}

We recall that $r = 2\pi |x + \alpha - i \beta|$, so as the Stokes line is crossed along the real axis
at ${\rm Re}(x) = -\alpha$ is given by $r = 2\pi\beta$. Consequently we see that $\mathcal{S}$ rapidly jumps from zero to $i \mathcal{S}_{\mathrm{on}}$ as the Stokes line is crossed, where $\mathcal{S}_{\mathrm{on}} =  -{4 \pi^3 \beta^2 F}/{h}$. The corresponding remainder contribution is given by
\begin{equation}
R_N(x) \sim i \mathcal{S}_{\mathrm{on}} e^{-2\pi i(x + \alpha-i \beta)}\qquad \mathrm{as} \qquad h \rightarrow 0.
\end{equation}

As before, we compute the remaining Stokes contributions and write this in terms of real-valued trigonometric functions, giving
\begin{equation}
R_N(x) \sim 2\mathcal{S}_1 e^{-2\pi \beta} \sin\left(\frac{2\pi(x-\alpha)}{h}\right) +
2\mathcal{S}_2 e^{-2\pi \beta} \sin\left(\frac{2\pi(x+\alpha)}{h}\right) ,\label{1.RNcos}
\end{equation}
where $\mathcal{S}_1$ switches rapidly as the Stokes line is crossed from zero to $\mathcal{S}_{\mathrm{on}}$ in a region of width $\mathcal{O}(h^{1/2})$ about the Stokes line $\mathrm{Re}(x) = -\alpha$. Similarly, $\mathcal{S}_2$ switches from zero to $\mathcal{S}_{\mathrm{on}}$ across the Stokes line $\mathrm{Re}(x) = \alpha$.
The exponentially small contributions are depicted very similar to the schematic picture
on Figure \ref{F:StokesFig1}. In the region ${\rm Re}(x) > \alpha$, we can rewrite \eqref{1.RNcos} as
\begin{equation}
R_{N} \sim- \frac{16 \pi^3 \beta^2 F}{h}\cos\left(\frac{2\pi\alpha}{h}\right)
\sin\left(\frac{2\pi x}{h}\right) ,\qquad \mathrm{as} \qquad \epsilon \rightarrow 0.\label{2.RTOT2}
\end{equation}
Therefore, $R_{N}$ cancels in this region $\mathrm{Re}(x) > \alpha$ if $2\pi\alpha/h = {\pi (2m-1)}/{2}$, where $m \in \mathbb{N}$.
If we denote these choices of the small parameter as $h_m$, this yields the asymptotic result \eqref{eps-advance-delay}.

\subsection{Numerical results}
\label{sec-4-3}

Here we approximate numerically on-site and inter-site lattice solitons (\ref{on-site-inter-site})
to the second-order difference equation (\ref{DiscreteU}). In accordance to Definition \ref{def-1}, we will
approximate the transparent points by using a computational method consistent with the one used in \cite{Melvin1},
where the transparent points were computed by finding the values of $h$,
for which the eigenvalue of the stability problem passes through zero. We detect
the transparent points by seeking for localized solution of linearized problem
\begin{equation}\label{Discrete}
\frac1{h^2} \left[ v_{n+1}-2v_n+v_{n-1} \right] + v_n-\frac{\theta
(1- (u_n)^2)}{(1+(u_n)^2)^2} v_n=0, \quad n \in \mathbb{Z},
\end{equation}
where $\{u_n\}_{n \in \mathbb{Z}} \in \ell^2(\mathbb{Z})$ is the solution of discrete equation (\ref{DiscreteU}).
If $\{v_n\}_{n \in \mathbb{Z}} \in \ell^2(\mathbb{Z})$ exists, then it corresponds to
the eigenvector of the stability problem with zero eigenvalue.
Cases of the on-site lattice soliton $\{u_n^{os}\}_{n \in \mathbb{Z}}$ and
the inter-site lattice soliton $\{u_n^{is}\}_{n \in \mathbb{Z}}$ are treated separately.

Let $\{u_n^{os}\}_{n \in \mathbb{Z}}$ be the on-site lattice soliton of the difference equation (\ref{DiscreteU})
for some $h$ and consider the linearized difference equation (\ref{Discrete}) with this $\{u_n^{os}\}_{n \in \mathbb{Z}}$.
The sequence $\{ v_n^{os} \}_{n \in \mathbb{Z}_-}$ satisfies the decay condition $v_n^{os} \to 0$ as $n \to -\infty$.
For large values of $n \to -\infty$, we can use the linear asymptotics
$v_n^{os}\sim C\gamma^n$, where $C > 0$ is constant and $\gamma$ is the root of dispersion equation
\begin{eqnarray}
\label{root-gamma}
\gamma^2-(2+(\theta-1)h^2)\gamma+1=0,
\end{eqnarray}
such that $|\gamma|>1$. Thanks to the symmetry of $\{u_n^{os}\}_{n \in \mathbb{Z}}$,
we require the eigenvector $\{v_n^{os}\}_{n \in \mathbb{Z}}$ to satisfy the same
symmetry as the translational (derivative) mode:
\begin{equation}
\label{sym-v}
v_{-n}^{os} = - v_n^{os}, \quad n \in \mathbb{Z}.
\end{equation}
Generically, the sequence $\{v_n^{os}\}_{n \in \mathbb{Z}_-}$ does not satisfy
the symmetry condition $v_0 = 0$ and hence violates the symmetry (\ref{sym-v}).
Moreover, if the sequence is continued to $n \in \mathbb{Z}_+$, it diverges generally as $n \to +\infty$.
Therefore, we introduce the function $W^{os}(h)=v_0^{os}$ and look for zeros of $W^{os}(h)$ as $h$ varies.

Similarly, let $\{u_n^{is}\}_{n \in \mathbb{Z}}$ be the inter-site lattice soliton of
the difference equation (\ref{DiscreteU}) for some $h$ and consider the linearized difference equation
(\ref{Discrete}) with this $\{u_n^{is}\}_{n \in \mathbb{Z}}$.
The sequence $\{ v_n^{is} \}_{n \in \mathbb{Z}_-}$ is computed by using
the same asymptotics $v_n^{is}\sim C\gamma^n$, where $C > 0$ is constant and $\gamma$ is the root of
Eq.~(\ref{root-gamma}) with $|\gamma| > 1$. Thanks to the symmetry of $\{u_n^{is}\}_{n \in \mathbb{Z}}$,
we require the eigenvector $\{v_n^{is}\}_{n \in \mathbb{Z}}$ to satisfy the same
symmetry as the translational (derivative) mode:
\begin{gather*}
\label{sym-v_is}
v_{-n}^{is} = - v_{n-1}^{is}, \quad n \in \mathbb{Z}.
\end{gather*}
Generically, a sequence $\{v_n^{is}\}_{n \in \mathbb{Z}_-}$ does not satisfy
the symmetry condition $v_{-1} + v_0 = 0$ and hence violates the symmetry (\ref{sym-v_is}). Again, we introduce
the function $W^{is}(h)=v_{-1}^{is} + v_0^{is}$ and look for zeros of $W^{is}(h)$ as $h$ varies.

\begin{figure}[h]
{\centerline{\includegraphics [scale=0.6]{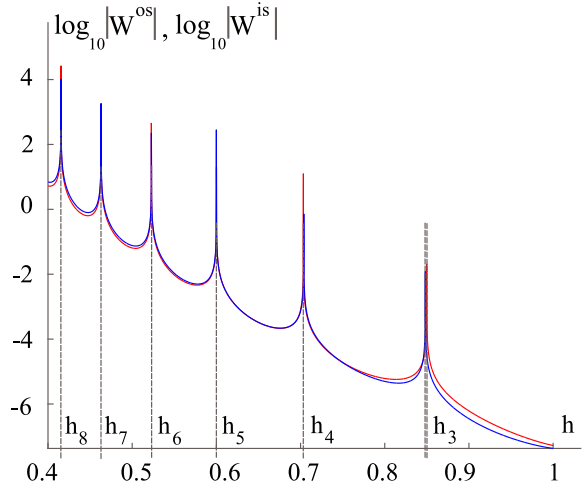}}} \caption{The plots of  $\log_{10}|W^{is}(h)|$ (red line) and $\log_{10}|W^{os}(h)|$ (blue line), $\theta=5$, $0.4<h<1$.  The peaks correspond to zeros of the functions $W^{is}(h)$ and $W^{is}(h)$.
Zeros $h_{3\div8}$ are shown, see enumeration of zeros in Table \ref{T1}} \label{oscillations}
\end{figure}

Our numerical procedure consists in computing the functions $W^{os}(h)$ and $W^{is}(h)$ with small enough spacing
with respect to $h$ and seeking for their zeros. If a continuous solution to
the advance-delay equation (\ref{advance-delay-intro}) exists at $h_*$, then
\begin{gather}\label{ZerosCoincide}
W^{os}(h_*)=W^{is}(h_*)=0.
\end{gather}
One transparent point is known at $h_* = \sqrt{2}$ for any $\theta>1$ thanks to the reduction
of the advance-delay equation (\ref{advance-delay-intro}) to the integrable Ablowitz--Ladik lattice \cite{Khare2005}.
This case was used for testing of the numerical procedure.

With this numerical algorithm for $\theta = 3,5,15$, we have obtained sequences $\{ h_m \}_{m \in \mathbb{N}}$
of zeros of $W^{os}(h)$ and $W^{is}(h)$ that are close to each other within the distance of $5\cdot 10^{-3}$.
The difference becomes even smaller and indistinguishable for zeros with smaller $h$ since
the step size in $h$ is smaller tan $10^{-5}$. We have also recovered the value $h_1 =\sqrt{2}$
at the first transparent point and found no zeros of $W^{os}(h)$ and $W^{is}(h)$ for $h \in (\sqrt{2},2.5)$.

\small
\begin{table}
\begin{tabular}{r||r|r||r|r||r|r}
$m$ & $h_m$, $\theta=3$    & $\Delta$ (on-off)         &$h_m$,  $\theta=5$ & $\Delta$ (on-off)             & $h_m$, $\theta=15$& $\Delta$ (on-off)\\[2mm]
\hline
 1&1.41421      & 0          &1.41421      & 0          &1.41421      &  0    \\[2mm]
 2&1.06639      & 0.00078    &1.06796     & 0.00480   &1.05366       &  0.02619\\[2mm]
 3&0.84855      & 0.00016    &0.84819     & 0.00160   & 0.83719       & 0.01749\\[2mm]
 4&0.70325      & 0          &0.70385     & 0.00048   & 0.69466       & 0.00924 \\[2mm]
 5&0.59979      & 0          &0.59999     & 0.00013   & 0.59257      & 0.00456\\[2mm]
 6&0.52263      & 0          &0.52282     & 0.00004   & 0.51625      & 0.00221 \\[2mm]
 7&0.46291      & 0          & 0.46303     & 0.00001   & 0.45165      & 0.00107\\[2mm]
 8&0.41536      & 0          & 0.41545     & 0           & 0.41011      & 0.00051\\[2mm]
 9&0.37662      & 0          & 0.37669     & 0           & 0.37180      & 0.00024\\[2mm]
10&0.34446      & 0          & 0.34452     & 0           & 0.34000      & 0.00012\\[2mm]
11&0.31733      & 0          & 0.31739     & 0           & 0.31320      & 0.00006 \\[2mm]
12&0.29415      & 0          & 0.29421     & 0           & 0.29030      & 0.00003\\[2mm]
13&0.27412      & 0          & 0.27417     & 0           & 0.27052      & 0.00001\\[2mm]
14&-            & -          & 0.25668     & 0           & 0.25325      & 0.00001\\[2mm]
15&-            & -          & 0.24128     & 0           & 0.23806      & 0\\[2mm]
16&-            & -          & 0.22762     & 0           & 0.22458      & 0\\[2mm]
17&-            & -          & 0.21533     & 0           & 0.21254      & 0\\[2mm]
18&-            & -          & 0.20566     & 0           & 0.20173      & 0\\[2mm]
19&-            & -          & -            & -          & 0.19196      & 0\\[2mm]
20&-            & -          & -            & -          & 0.18309      & 0\\[2mm]
21&-            & -          & -            & -          & 0.17501      & 0\\[2mm]
22&-            & -          & -            & -          & 0.16760      & 0\\[2mm]
23&-            & -          & -            & -          & 0.16080      & 0\\[2mm]
24&-            & -          & -            & -          & 0.15453      & 0%
\end{tabular}
\caption{Transparent points for the advance-delay equation (\ref{advance-delay-intro}).
The results are presented for $\theta=3,5,15$. For each value of $\theta$ two entries are shown:
average $h_n$ between the zeros of $W^{os}(h)$ and $W^{is}(h)$ (the 2-nd, 4-th and 6-th columns)
and the distance $\Delta$ between them (3-rd, 5-th and 7-th columns).
The first zero corresponds to the exact value $h_1=\sqrt{2}$.}\label{T1}
\end{table}
\normalsize

Table \ref{T1} represents these numerical results. It is interesting that zeros of  $W^{os}(h)$ and $W^{is}(h)$
change insignificantly for different values of $\theta$. For instance, the 13-th zero in Table \ref{T1}
differs by 1\% between $\theta=3$ to $\theta=15$. This fact is explained by slow dependence of $\alpha$ 
from parameter $\theta$. Indeed, $\alpha \approx 2.2025$ for $\theta = 3$ and $\alpha \approx 1.9771$ for $\theta = 15$. 

The plots of $\log_{10}|W^{os}(h)|$  and $\log_{10}|W^{is}(h)|$ for $\theta=5$ and for $0.4<h<1$
are shown in Fig.~\ref{oscillations}. The peaks correspond to zeros of the functions $W^{os}(h)$ and $W^{is}(h)$.
They are consistent with the values in Table \ref{T1}. Note that the difference is visible
for $h_2$, $h_3$ but becomes negligible for $h_4$ and smaller values of $h$.

\begin{figure}[h]
{\centerline{\includegraphics [scale=0.6]{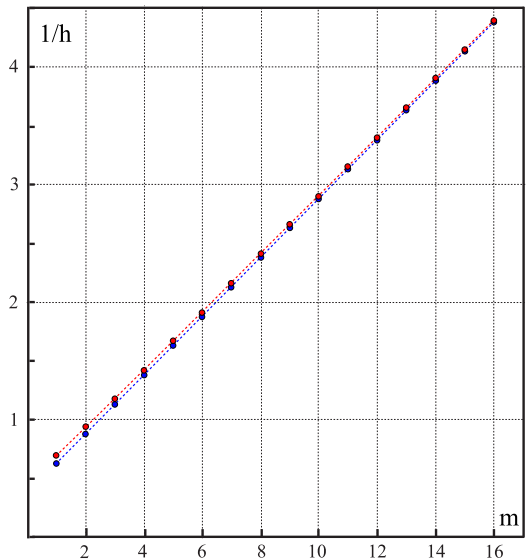}}} \caption{The numerical values of $h_m^{-1}$
corresponding to the transparent points for the advance-delay equation (\ref{advance-delay-intro}) (red balls)
and their asymptotic predictions (blue balls)  computed by the formula (\ref{spacing-advance-delay}) for $m=3\div 16$.
Here $\theta=5$ and the numerical values are taken from Table \ref{T1}.}
\label{C_Theta_5_AD}
\end{figure}

Fig.~\ref{C_Theta_5_AD} shows the numerical values of $\{h_m\}_{m \in \mathbb{N}}$ (red balls)
and their asymptotical values computed by the formula (\ref{spacing-advance-delay}) (blue balls)
for $m=3\div 16$ and $\theta=5$. The numerical values are taken from Table \ref{T1}.
The numerical results are in excellent agreement with the asymptotic formula.

\section{Conclusion}

We have addressed the existence of transparent points for standing lattice solitons in the
dNLS model with saturation and presented three groups of results.

Rigorous results are derived on existence and analytical continuation of solitary wave solutions
to the second-order differential equation which corresponds to the continuum limit.
By studying analytic mappings, we proved existence of a quadruple of logarithmic branch point
singularities in the complex plane nearest to the real line with a specific analytic behaviour
near the singularities.

These rigorous results are used in the asymptotic computations supporting our conjecture on existence
of an infinite countable set of solitary waves in the fourth-order differential equation which
corresponds to the next-order in the continuum limit. We presented two alternative asymptotic
computations producing identical results: one relies on computations of oscillatory integrals
in the persistence problem and the other one relies on beyond-all-order theory.
With application of these results to the advance-delay equation, we can only conjecture
on existence of an infinite countable set of transparent points for which the standing lattice solitons
are nearly continuous.

Finally, careful numerical computations are performed to show validity of our asymptotic predictions.
Numerical computations of solitary wave solutions in the fourth-order
differential equation agree well with the asymptotic formula. Numerical computations of
standing lattice solitons also confirmed existence of the countable set of transparent points.

\vspace{0.25cm}

{\bf Acknowledgments:} GLA was funded by  Russian Science Foundation (Grant No. 17-11-01004).
DEP acknowledges a financial support from the State task program in the sphere
of scientific activity of Ministry of Education and Science of the Russian Federation
(Task No. 5.5176.2017/8.9) and from the grant of President of Russian Federation
for the leading scientific schools (NSH-2685.2018.5).


\begin{thebibliography}{99}

\bibitem{prl-alfimov} G.L. Alfimov, E.V. Medvedeva, and D.E. Pelinovsky,
``Wave systems with an infinite number of localized travelling waves",  Phys. Rev. Lett. {\bf 112} (2014), 054103 (5 pages).

\bibitem{Bennett1}
T. Bennett, C. J. Howls, G. Nemes and A. B. Olde Daalhuis, ``Globally exact asymptotics for integrals with arbitrary order saddles",
SIAM J. Math. Anal. {\bf 50} (2018) 2144--2177.

\bibitem{Berry0} M. V. Berry, "Stokes' phenomenon; smoothing a {V}ictorian discontinuity", Pub. Math. de l'IHÉS {\bf 68} (1988) 211--221.


\bibitem{Berry1} M. V. Berry and C. J. Howls, ``Hyperasymptotics",
Proc. Roy. Soc. Lond. A. {\bf 430} (1990) 653--668.

\bibitem{Berry2} M. V. Berry and C. J. Howls, ``Hyperasymptotics for integrals with saddles",
Proc. Roy. Soc. Lond. A. {\bf 434} (1991) 657--675.


\bibitem{Boyd2} J. P. Boyd, ``The Devils Invention: {A}symptotic, Superasymptotic and Hyperasymptotic Series",
Acta Appl. Math. {\bf 56} (1999) 1--98.

\bibitem{Chong1} R. Carretero-Gonz\'ales, J. D. Talley, C. Chong and B. A. Malomed, ``Multistable solitons in
the cubic--quintic discrete nonlinear Schr\"{o}dinger equation", Physica D {\bf 216} (2006), 77--89.

\bibitem{Chapman1} S. J. Chapman, J. R. King, J. R. Ockendon, K. L. Adams,
``Exponential asymptotics and {S}tokes lines in nonlinear ordinary differential equations",
Proc. Roy. Soc. Lond. A. {\bf 454} (1998) 2733--2755.

\bibitem{Chong2} C. Chong and D.E. Pelinovsky, ``Variational approximations of bifurcations
of asymmetric solitons in cubic--quintic nonlinear Schr\"{o}dinger lattices",
DCDS S {\bf 4} (2011), 1019--1031.

\bibitem{Dingle1} R. B. Dingle, {Asymptotic Expansions: {T}heir Derivation and Interpretation}
(Academic Press, New York, 1973).

\bibitem{Daalhuis1}
A. B. Olde Daalhuis, S. J. Chapman, J. R. King, J. R. Ockendon and R. H. Tew,
``Stokes phenomenon and matched asymptotic expansions",
SIAM J. Appl. Math. {\bf 55} (1995), 1469--1483.

\bibitem{dmitriev1}
S.V. Dmitriev, P.G. Kevrekidis, N. Yoshikawa, and D.J. Frantzeskakis,
``Exact stationary solutions for the translationally invariant discrete nonlinear Schr\"{o}dinger equations",
J. Phys. A: Math. Theor. {\bf 40} (2007), 1727--1746.

\bibitem{dmitriev2} S.V. Dmitriev, P.G. Kevrekidis, A.A. Sukhorukov, N. Yoshikawa, and
S. Takeno, ``Discrete nonlinear Schr\"{o}dinger equations free of the Peierls--Nabarro potential",
Phys. Lett. A {\bf 356} (2006) 324--332.

\bibitem{Fedoryuk}
M.V. Fedoryuk, {\em The Saddle-point Method} (Moscow, Nauka, 1977) (In Russian)

\bibitem{GH1991} S. Gatz and J. Herrmann, ``Soliton propagation in materials with saturable nonlinearity",
J. Opt. Soc. Am. B {\bf 8} (1991), 2296--2302.

\bibitem{Gel1} V.G. Gelfreich, V.F. Lazutkin, and M. B. Tabanov, ``Exponentially small splittings in Hamiltonian systems",
Chaos {\bf 1}  (1991), 137--142.

\bibitem{Gel2} V. Gelfreich and C. Simo, ``High-precision computations of divergent asymptotic series and homoclinic phenomena",
Discrete Contin. Dyn. Syst. Ser. B {\bf 10}  (2008), 511--536.

\bibitem{GrimshawJoshi} R. Grimshaw and N. Joshi, ``Weakly nonlocal solitary waves in
a singularly perturbed Korteweg--de Vries equation", SIAM J. Appl. Math. {\bf 55} (1995), 124--135.

\bibitem{Maulskov} L. Hadzievski, A. Maluckov, M. Stepic, and D. Kip, ``Power controlled soliton stability
and steering in lattices with saturable nonlinearity", Phys. Rev. Lett. {\bf 93} (2004), 033901 (4 pp).

\bibitem{Hoffman} A. Hoffman and J.D. Wright, ``Nanopteron solutions of diatomic Fermi--Pasta--Ulam--Tsingou lattice
with small mass-ratio", Physica D {\bf 358} (2017), 33--59.

\bibitem{Howls1} C. J. Howls, ``Hyperasymptotics for Integrals with Finite Endpoints",
Proc. Math. Phys. Sci. {\bf 439} (1992) 373--396.

\bibitem{Howls2} C. J. Howls, ``Hyperasymptotics for multidimensional integrals, exact remainder terms and the global connection problem",
Proc. Roy. Soc. Lond. A. {\bf 453} (1997) 2271--2294.

\bibitem{HPS} H.J.Hupkes, D.E. Pelinovsky, and B. Sandstede, ``Propagation failure in the discrete Nagumo equation",
Proc. AMS {\bf 139} (2011), 3537--3551.

\bibitem{Iooss} G. Iooss and D.E. Pelinovsky, ``Normal form for travelling kinks in discrete Klein-Gordon lattices",
Physica D {\bf 216} (2006), 327--345.

\bibitem{Joshi1} N. Joshi and C. J. Lustri, ``Stokes phenomena in discrete Painlev{\'e} I", Proc. R. Soc. \rm{A}.
{\bf 471} (2015) 20140874 (22 pages).

\bibitem{Joshi2} N. Joshi, C. J. Lustri and S. Luu, ``Stokes phenomena in discrete Painlev{\'e} II",
Proc. R. Soc. \rm{A}. {\bf 473} (2017), 20160539 (20 pages).

\bibitem{panos-book} P.G. Kevrekidis,
{\it Discrete Nonlinear Schrodinger Equation: Mathematical Analysis, Numerical Computations and Physical Perspectives},
(Springer-Verlag, Berlin, 2009).

\bibitem{Khare2005} A. Khare, K.O. Rasmussen, M.R. Samuelsen, and A. Saxena, ``Exact solutions of the saturable
discrete nonlinear Schr\"{o}dinger equation", J. Phys. A: Math. Gen. {\bf 38} (2005), 807--814.

\bibitem{King} J. R. King and S. J. Chapman, ``Asymptotics beyond all orders and Stokes lines in nonlinear differential-difference equations",
Eur. J. Appl. Math. {\bf 4} (2001), 433--463.

\bibitem{Lustri} C. Lustri and M.A. Porter, ``Nanoptera in a period-2 Toda chain", SIAM J. Appl. Dynam. Syst.
{\bf 17}  (2018),  1182--1212.

\bibitem{Melvin1} T.R.O. Melvin, A.R. Champneys, P.G. Kevrekidis, and J. Cuevas, ``Radiationless traveling waves
in saturdable nonlinear Schr\"{o}dinger lattices", Phys. Rev. Lett. {\bf 97} (2006), 124101 (4 pages)

\bibitem{Melvin2} T.R.O. Melvin, A.R. Champneys, P.G. Kevrekidis, and J. Cuevas, ``Travelling solitary waves
in the discrete nonlinear Schr\"{o}dinger equation with saturable nonlinearity: existence, stability and dynamics",
Physica D {\bf 237} (2008), 551--567.

\bibitem{CMP2009} T.R.O. Melvin, A.R. Champneys, and D.E. Pelinovsky, ``Discrete traveling solitons in
the Salerno model", SIAM J. Appl. Dynam. Systems {\bf 8} (2009), 689--709.


\bibitem{barashenkov2007} O.F. Oxtoby and I.V. Barashenkov, ``Moving solitons in the discrete
nonlinear Schr\"{o}dinger equation", Phys. Rev. E {\bf 76} (2007), 036603.

\bibitem{pel1} D.E. Pelinovsky, ``Translationally invariant nonlinear Schr\"{o}dinger lattices", Nonlinearity
{\bf 19} (2006), 2695--2716.

\bibitem{P11} D.E. Pelinovsky, ``Traveling monotonic fronts in the discrete Nagumo equation",
J. Dynam. Diff. Eqs. {\bf 23} (2011), 167--183.

\bibitem{pelin-book} D.E. Pelinovsky, {\em Localization in periodic potentials: from Schr\"{o}dinger operators to the
Gross--Pitaevskii equation}, LMS Lecture Note Series {\bf 390} (Cambridge University Press, Cambridge, 2011).

\bibitem{pel2} D.E. Pelinovsky, T.R.O. Melvin, and A.R. Champneys, ``One-parameter localized traveling waves
in nonlinear Schr\"{o}dinger lattices", Physica D {\bf 236} (2007), 22--43.

\bibitem{PelRothos} D.E. Pelinovsky and V.M. Rothos, ``Bifurcations of travelling breathers in the discrete NLS equations",
Physica D {\bf 202} (2005) 16--36.

\bibitem{PRG} Y. Pomeau, A. Ramani, and B. Grammaticos, ``Structural stability of the Korteweg–de Vries solitons
under a singular perturbation", Physica D {\bf 31} (1988), 127--134.

\bibitem{QinXiao} W.-X. Qin, X. Xiao, ``Homoclinic orbits and localized
solutions in nonlinear Schr\"{o}dinger lattices", Nonlinearity {\bf 20} (2007), 2305--2317.

\bibitem{Syafwan} M. Syafwan, H. Susanto, S.M. Cox, and B.A. Malomed,
``Variational approximations for traveling solitons in a
discrete nonlinear Schr\"{o}odinger equation", J. Phys. A: Math. Theor. {\bf 45} (2012) 075207 (18pp).

\bibitem{Dawes} C. Taylor and J.H.P. Dawes, ``Snaking and isolas of localised states in bistable discrete lattices",
Phys. Lett. A {\bf 375} (2010), 4968--4976.

\bibitem{Tovbis1} A. Tovbis, ``Breaking homoclinic connections for a singularly perturbed differential equation and the
Stokes phenomenon", Stud. Appl. Math. {\bf 104} (2000), 353--386.

\bibitem{Tovbis2} A. Tovbis and D. Pelinovsky, ``Exact conditions for existence of homoclinic orbits in the fifth-order KdV model",
Nonlinearity {\bf 19} (2006), 2277--2312.

\bibitem{Tovbis} A. Tovbis, M. Tsuchiya, and C. Jaffe, ``Exponential asymptotic expansions and approximations
of the unstable and stable manifolds of singularly perturbed systems with the Henon map as an example",
Chaos {\bf 8} (1998), 665--681.

\bibitem{VF92} A. Vanderbaumwhede and B. Fiedler, ``Homoclinic period blow-up in reversible and conservative systems",
Z.Angew. Math. Phys. {\bf 43} (1992), 292--318.

\bibitem{Volpert} V. Vougalter and V. Volpert, ``Solvability conditions for some non-Fredholm operators",
Proc. Edinb. Math. Soc. {\bf 54}  (2011), 249--271.

\bibitem{Volpert2} V. Vougalter and V. Volpert, ``Solvability conditions for some linear and nonlinear non-Fredholm elliptic problems",
Anal. Math. Phys. {\bf 2}  (2012), 473--496.

\bibitem{Vainstein} A. Vainchtein, Y. Starosvetsky, J.D. Wright, and R. Perline, ``Solitary waves
in diatomic chains", Phys. Rev. E {\bf 93} (2016), 042210

\end{thebibliography}
\end{document}